\documentclass[11pt]{article}
\usepackage[margin=0.8in]{geometry}

\usepackage{varioref}
\usepackage[usenames,svgnames,xcdraw,table]{xcolor}
\definecolor{DarkBlue}{rgb}{0.1,0.1,0.5}
\definecolor{DarkGreen}{rgb}{0.1,0.5,0.1}

\usepackage{nicefrac}
\usepackage{cancel}

\usepackage{hyperref}
\hypersetup{
   colorlinks   = true,
 linkcolor    = DarkBlue, 
 urlcolor     = DarkBlue, 
	 citecolor    = DarkGreen 
}

\usepackage{appendix}

\usepackage{mathtools}
\usepackage{amsthm}
\usepackage{algorithmic}
\usepackage{algorithm}
\usepackage{array}

\usepackage[font={small,it}]{caption}

\usepackage{graphicx,epsfig,amsmath,latexsym,amssymb,verbatim}

\usepackage[square,semicolon,numbers]{natbib}

\newcommand{\extra}[1]{}

\usepackage{amsmath,amssymb,amsfonts}
\usepackage{amsthm}
\usepackage{thmtools,thm-restate}
\usepackage{enumitem}

\newtheorem{theorem}{Theorem}
\newtheorem{corollary}[theorem]{Corollary}
\newtheorem{definition}{Definition}
\newtheorem{lemma}[theorem]{Lemma}
\newtheorem{proposition}{Proposition}
\newtheorem{claim}{Claim}

\def\squareforqed{\hbox{\rlap{$\sqcap$}$\sqcup$}}
\def\qed{\ifmmode\squareforqed\else{\unskip\nobreak\hfil
\penalty50\hskip1em\null\nobreak\hfil\squareforqed
\parfillskip=0pt\finalhyphendemerits=0\endgraf}\fi}
\def\endenv{\ifmmode\;\else{\unskip\nobreak\hfil
\penalty50\hskip1em\null\nobreak\hfil\;
\parfillskip=0pt\finalhyphendemerits=0\endgraf}\fi}

\renewenvironment{proof}{\noindent \textbf{{Proof~} }}{\qed\medskip}
\newenvironment{proof+}[1]{\noindent \textbf{{Proof #1~} }}{\qed\medskip}

\mathchardef\ordinarycolon\mathcode`\:
\mathcode`\:=\string"8000
\def\vcentcolon{\mathrel{\mathop\ordinarycolon}}
\begingroup \catcode`\:=\active
  \lowercase{\endgroup
  \let :\vcentcolon
  }

\DeclareMathOperator*{\argmin}{arg\,min}

\usepackage{physics}
\usepackage[colorinlistoftodos,prependcaption,textsize=tiny]{todonotes}

\usepackage{thm-restate,amsthm,thmtools}
\usepackage{titlesec}
\usepackage{fancyhdr}
\usepackage{verbatim}
\usepackage{tabularx}

\newcommand{\MMS}{\mathrm{MMS}}

\newcommand{\Alloc}{\mathcal{A}}
\newcommand{\AllocB}{\mathcal{B}}
\newcommand{\AllocX}{\mathcal{X}}

\newcommand{\PO}{\mathrm{PO}}

\newcommand{\EFone}{\mathrm{EF}1}

\newcommand{\EFX}{\mathrm{EFX}}
\newcommand{\EFkX}{\mathrm{EFkX}}
\newcommand{\SC}{\mathrm{SC}}
\newcommand{\umat}{\widehat{\mathcal{M}}}
\newcommand{\mat}{\mathcal{M}}

\newcommand\numberthis{\addtocounter{equation}{1}\tag{\theequation}}

\title{\bfseries Fair Chore Division under Binary Supermodular Costs}

\author{Siddharth Barman\thanks{Indian Institute of Science. {\tt barman@iisc.ac.in}} \and Vishnu V. Narayan\thanks{Tel Aviv University. {\tt narayanv@tauex.tau.ac.il}} \and Paritosh Verma\thanks{Purdue University. {\tt paritoshverma97@gmail.com}}}

\date{}

\begin{document}

\maketitle

\begin{abstract}
We study the problem of dividing indivisible chores among agents whose costs (for the chores) are supermodular set functions with binary marginals. Such functions capture complementarity among chores, i.e., they constitute an expressive class wherein the marginal disutility of each chore is either one or zero, and the marginals increase with respect to supersets. In this setting, we study the broad landscape of finding fair and efficient chore allocations. In particular, we establish the existence of $(i)$ $\EFone$ and Pareto efficient chore allocations, $(ii)$ $\MMS$-fair and Pareto efficient allocations, and $(iii)$ Lorenz dominating chore allocations. Furthermore, we develop polynomial-time algorithms---in the value oracle model---for computing the chore allocations for each of these fairness and efficiency criteria. Complementing these existential and algorithmic results, we show that in this chore division setting, the aforementioned fairness notions, namely $\EFone$, $\MMS$, and Lorenz domination are incomparable: an allocation that satisfies any one of these notions does not necessarily satisfy the others.

Additionally, we study $\EFX$ chore division. In contrast to the above-mentioned positive results, we show that, for binary supermodular costs, Pareto efficient allocations that are even approximately $\EFX$ do not exist, for \emph{any} arbitrarily small approximation constant. Focusing on EFX fairness alone, when the cost functions are identical we present an algorithm (Add-and-Fix) that computes an EFX allocation. For binary marginals, we show that Add-and-Fix runs in polynomial time.
\end{abstract}

\section{Introduction}

The question of dividing indivisible items among a set of agents in a fair manner is a pervasive problem in many domains. Popular notions of fairness in the field of discrete fair division include envy-freeness up to one good ($\EFone$), envy-freeness up to any good ($\EFX$), and the {maximin share guarantee} ($\MMS$). Both existential and algorithmic guarantees for these and related fairness notions have been extensively studied in recent years; see e.g., \cite{aziz2022algorithmic,endriss2018lecture}. 

In this research direction, a majority of results focus on the fair division of \textit{goods}, which correspond to items that, when allocated, induce non-negative values among the agents. Notably, the complementary settings of fair division of \emph{chores} (which model negatively valued items or tasks) are relatively under-explored. While the definitions of familiar fairness criteria (such as $\EFone$ and $\MMS$) extend quite directly, the conditions under which a fair chore division exists do not mirror the goods' case. 

In fact, important known results for the goods setting do not directly extend to the chores setup. For example, the influential work of Caragiannis et al.~\cite{caragiannis2019unreasonable} studies the fair division of goods when the agents have additive valuations and establishes that, in this context, there always exists an allocation that is both $\EFone$ and Pareto efficient. In particular, they show that an allocation of goods that maximizes Nash welfare among the agents achieves these fairness and efficiency goals. By contrast, in the context of chores, the existence of allocations that are simultaneously $\EFone$ and Pareto efficient has remained a challenging open problem. Such allocations have only recently been shown to exist for a specific subclass of valuation functions, namely \textit{bivalued} additive valuations~\cite{ebadian2022,garg2022fair}. Consequently, it is clear that the fair division of chores presents a new set of technical challenges, and the study of chore division is an important thread of research in discrete fair division.

We contribute to the recent literature on chore division by focusing on settings wherein the agents' costs (disutilities) for the chores have \textit{binary marginals}. Specifically, an agent $i$'s cost function, $c_i$, is said to have binary marginals (equivalently, is said to be dichotomous) if the marginal value of the chore $t$ relative to any subset $S$ is either zero or one, i.e., $c_i(S \cup \{t \}) - c_i(S) \in \{0,1\}$. In the complementary context of goods, an abundance of papers consider agents with dichotomous valuations (e.g.~\cite{bogomolnaia2005collective,kurokawa2018leximin}), since such valuations model agent preferences in several real-world settings, such as kidney exchanges \cite{roth2005pairwise} and housing allocations \cite{benabbou2020finding}. 

The majority of our results focus on the case where the agents have \textit{supermodular} cost functions. Supermodular functions have received considerable attention in the economics literature. Notably, the use of supermodularity to express complementarity in agents' preferences dates back to the works of Edgeworth, Pareto, and Fisher~\cite{samuelson1948foundations}.\footnote{While we focus on chore division under supermodular costs, prior works have complementarily addressed fair division of goods with supermodular valuations; see, e.g.,~\cite{caragiannis2019unreasonable,benabbou2020finding}.} Specifically, an agent $i$'s cost function $c_i$ (i.e., disutilities for the chores) is said to be supermodular iff it bears the increasing marginals property: $c_i(T  \cup \{a \}) - c_i(T) \geq c_i(S  \cup \{a \}) - c_i(S)$, for all subsets $S \subseteq T$ and all chores $a \notin T$. Increasing marginals are a well-suited assumption for chores, since taking on a new task is increasingly likely to raise one's cost due to the burdens of multitasking and frequent task-switching. Binary supermodular functions can be used to model, for instance, the costs associated with page caching, or contexts in which only the first few trials of some software (or delivery service) are free. \\

\subsection{Our Results} 
We study the problem of finding fair and efficient allocations of indivisible chores among agents with binary supermodular cost functions. Our work develops several results on the existence and computability of $\EFone$, $\MMS$, Pareto efficient ($\PO$), and Lorenz dominating allocations. Specifically, we show that, for these cost functions, (i) an allocation that is $\EFone$ and $\PO$, (ii) an allocation that is $\PO$ and in which every agent receives its minimax share ($\MMS$), and (iii) a Lorenz dominating allocation, always exist and can be computed in polynomial time (given value-oracle access to the cost functions). These results constitute some of the first positive chore-division guarantees for standard fairness and economic-efficiency notions in discrete fair division. 

In the current context of chore division with binary supermodular costs, we also show that certain pairs of these guarantees, such as Lorenz domination and $\EFone$ (or $\MMS$ along with $\EFone$), are \textit{incomparable}, i.e., an allocation that satisfies one of these fairness criterion does not necessarily bear the other. This is in contrast to some well-known results for goods. For instance, in the complementary case of goods division with binary submodular valuations, {\it any} Lorenz dominating allocation is $\EFone$~\cite{babaioff2021fair}. 
It is interesting to note that while some positive results carry forward from the goods setting to the case of chores, others are negated. 

Our work identifies binary supermodular costs as a relevant function class for which $\EFone$ and $\PO$ allocations of chores are guaranteed to exist. Given this result, a natural follow-up question is whether this guarantee can be strengthened: do $\EFX$ and $\PO$ allocations always exist in the current context? We answer this question in the negative by showing a significantly stronger negative result: for binary supermodular costs, $\PO$ and $\beta$-$\EFkX$ allocations do not exist for any $\beta\in(0,1]$ and for any $k\geq1$, even when the cost functions are identical.

Complementing this negative result and focusing on fairness alone, we present positive results towards the existence of $\EFX$ for chores. We show algorithmically that when the agents have \textit{identical} cost functions, an $\EFX$ allocation always exists. Notably, this result only requires the (identical) cost function to be monotonic. Our algorithm, that we call Add-and-Fix, provides an alternate proof of the existence of $\EFX$ chore allocations for the identical valuations case. A result of Plaut and Roughgarden~\cite{plaut2020almost} shows that a leximin-type solution (that they call leximin++) is $\EFX$ for identical valuations in the case of goods. It is straightforward to extend this result to show the existence of $\EFX$ for chores (see, e.g.,~\cite{aleksandrov2018almost}). However, the leximin++ solution is {\rm NP}-hard to compute even when the agents have identical costs with binary marginals (see Appendix~\ref{appendix:leximax-hardness}). By contrast, our algorithm (Add-and-Fix) obtains an $\EFX$ allocation in polynomial time in this setting. In fact, for any monotonic cost function $c(\cdot)$ that is integer-valued (i.e., $c(S) \in \mathbb{Z}_{\geq 0}$ for all subsets $S$), Add-and-Fix runs in pseudo-polynomial time. Consequently, for the particular case of identical cost functions with binary marginals, we obtain a polynomial-time algorithm for finding $\EFX$ chore allocations. 


\subsection{Additional Related Work} 
Over the past decades, the fair division problem has emerged as a central and influential topic at the interface of mathematical economics and computer science. A collection of early fair-division results study the problem of achieving \textit{envy-freeness} \cite{Fol67}, where no agent prefers the bundle of another. Most of these classic works focus on the divisible setting, where a heterogeneous item (a \textit{cake}) can be fractionally divided to create an allocation; see e.g.~\cite{Str80,brams1996fair}. \\

\noindent 
\textbf{\textit{Indivisible items.}} Over the preceding few years, the research focus has evolved towards studying the fair division of \textit{indivisible goods}, where each good has to be allocated integrally to one agent. Since envy-freeness cannot be necessarily achieved in such combinatorial settings, the goal here is to obtain any of a variety of its relaxations and other approximate fairness criteria, including envy-freeness up to one item ($\EFone$) \cite{budish2011combinatorial,lipton2004approximately}, envy-freeness up to any item ($\EFX$)~\cite{caragiannis2019unreasonable,plaut2020almost}, and the maximin share ($\MMS$) guarantee~\cite{budish2011combinatorial,kurokawa2018fair}. \\

\noindent
\textbf{\textit{Pareto efficiency.}} While most of the results mentioned above consider the problem of achieving fairness alone, a central desideratum is to seek fair allocations that are also economically efficient. 
A standard notion of efficiency in mathematical economics is that of Pareto efficiency or Pareto optimality ($\PO$), in which no agent can be made strictly better off without making at least one other agent worse off in the process. As mentioned previously, the work of Caragiannis et al.~\cite{caragiannis2019unreasonable} shows that when the agents have additive valuations over the goods, an allocation that maximizes Nash welfare is simultaneously $\EFone$ and $\PO$. However, computing a Nash welfare maximizing allocation is ${\rm NP}$-hard. Barman et al.~\cite{barman2018finding} bypass this hardness by showing that an $\EFone$ and $\PO$ allocation (of goods) can be directly computed in pseudo-polynomial time for agents with additive valuations. \\

\noindent
\textbf{\textit{Binary marginals.}} Binary marginals have received substantial attention in the fair division literature for goods (see e.g.~\cite{bogomolnaia2004random,kurokawa2018leximin,babaioff2021fair}). Darmann and Schauer~\cite{darmann2015maximizing} develop an efficient algorithm to find a Nash Welfare maximizing allocation for binary additive valuations, while Barman and Verma~\cite{BP21} present an approximation algorithm for the same problem for the more general class of binary XOS valuations. Truthful mechanisms for fair division of goods under binary additive valuations \cite{halpern2020fair} and binary submodular valuations~\cite{babaioff2021fair} have also been developed in prior works. Babaioff et al.~\cite{babaioff2021fair} and Benabbou et al.~\cite{benabbou2020finding} show the existence and polynomial-time computability of allocations that are $\EFX$ and $\PO$ (and, hence, $\EFone$ and $\PO$) for the class of binary submodular (or, equivalently, matroid-rank) valuations. For this valuation class, Babaioff et al.~\cite{babaioff2021fair} establish existence, efficient computation, and fairness implications of Lorenz dominating allocations. \\

\noindent
\textbf{\textit{Chore division.}}  
While the definitions of envy-freeness, $\MMS$ and Pareto optimality extend directly from the goods case to the chores setting, the fairness criteria of $\EFone$ and $\EFX$ are typically defined via the removal of a chore from the envying agent's bundle (rather than a good from the envied agent's bundle). For approximately-$\MMS$-fair division of chores, several results are known, including an $\frac{11}{9}$-approximation factor \cite{huang2021algorithmic} and a $\frac{44}{43}$-impossibility bound \cite{feige2021tight} under additive valuations. As mentioned previously, while allocations that are $\EFone$ and $\PO$ are known to exist in a variety of settings for goods, in the chore division case such existence results are only known for bivalued instances \cite{ebadian2022,garg2022fair}. We refer the reader to the recent survey by Amanatidis et al.~\cite{amanatidis2022fair} for a comprehensive overview of the discrete fair division literature.\\

\noindent
{\bf Organization of This Paper.} 
Section~\ref{sec:prelim} contains the formal model and relevant definitions, while Sections~\ref{sec:cost-min} through 
\ref{sec:efx} contain our technical results.

First, in Section~\ref{sec:cost-min}, we study the structure of binary supermodular set functions. We present a useful connection between binary supermodular cost functions and the well-known matroid rank valuation functions for goods (Lemma~\ref{lemma:submod-supermod}). We also show (Lemmas~\ref{lemma:min-cost-characterization}~and~\ref{lemma:completing-the-partial-alloc} along with Theorem~\ref{theorem:social-cost-min}) that a variety of cost-minimizing allocations can be computed in polynomial time; these results will be helpful in the subsequent sections and may be of independent interest.

In Section~\ref{sec:positive}, we present our positive results on fair and efficient chore division with binary supermodular cost functions. Specifically, we establish the existence and computability of chore allocations that are simultaneously social-cost-minimizing and satisfy various fairness criteria, under binary super modular costs (Theorems~\ref{thm:po-ef1}, \ref{thm:po-mms}, and \ref{thm:lorenz}). In Section~\ref{sec:incom}, we show  that these fairness properties are incomparable, and in some cases, incompatible (Theorem~\ref{thm:incomparable}). In Section~\ref{sec:efx}, we study the EFX fairness notion. We show that for binary supermodular costs, Pareto efficient alocations that are even approximately EFX do not exist. Focusing on fairness alone, we present an algorithm (Add-and-Fix) that computes an EFX allocation under identical monotone cost functions.


\section{Notations and Preliminaries} \label{sec:prelim}

We consider the problem of dividing a set of $m$ indivisible chores among $n$ agents. Throughout, the sets of chores and agents will be, respectively, denoted by $[m] = \{1, 2, \ldots, m\}$ and $[n] = \{1, 2, \ldots, n\}$. An allocation $\Alloc = (A_1, A_2, \ldots, A_n)$ is a partition of the set of chores $[m]$ into $n$ disjoint subsets, i.e., $A_i \cap A_j = \emptyset$, for all $i \neq j$, and  $\cup_{i \in [n]} A_i = [m]$. In an allocation $\Alloc = (A_1, A_2, \ldots, A_n)$, the subset of chores (also referred as a \emph{bundle}) $A_i$ is assigned to agent $i$. We call an allocation $\Alloc = (A_1, A_2, \ldots, A_n)$ \emph{complete} if all the $m$ chores have been allocated, $\cup_{i \in [n]} A_i = [m]$. Also, an allocation $\Alloc$ is said to be \emph{partial} if it is not complete, $\cup_{i \in [n]} A_i \subsetneq [m]$. The set of all complete allocations of the chores $[m]$ among the $n$ agents will be denoted by $\Pi_n([m]) \coloneqq \{ (A_1, A_2, \ldots, A_n) \ : \ \cup_{i \in [n]} A_i = [m] \text{ and } A_i \cap A_j = \emptyset $ for each $ i \neq j \in [n] \}$. Henceforth, we will simply use the term allocation to denote a complete allocation. For notational convenience, we will write $S \cup \{a\}$ as $S + a$ and $S \setminus \{a\}$ as $S - a$, for any subset $S \subseteq [m]$ and chore $a \in [m]$. \\

\noindent
{\bf Cost functions.} For each agent $i \in [n]$, the cardinal disutilies for subsets of chores is specified by a cost function $c_i : 2^{[m]} \mapsto \mathbb{R}_{\geq 0}$. In particular, for any subset of chores $S \subseteq [m]$, agent $i$'s cost (disutility) for $S$ is denoted as $c_i(S) \in \mathbb{R}_{\geq 0}$. We will write $c_i(\{j\})$ as $c_i(j)$. A chore division instance is formally defined as a triple $\langle [n], [m], \{c_i\}_{i \in [n]} \rangle$. 

We focus on chore division instances in which the cost function $c_i$ of each agent $i \in [n]$ satisfies the \emph{binary marginals} property. Specifically, a set function $f: 2^{[m]} \mapsto \mathbb{R}_{\geq 0}$ bears  the binary marginals property iff, for every subset $S \subseteq [m]$ and chore $a \in [m]$, we have $f(S + a) - f(S) \in \{0,1\}$. That is, the marginal increase in cost upon adding a chore to a bundle is either $0$ or $1$. Note that, by definition, cost functions $f: 2^{[m]} \mapsto \mathbb{R}_{\geq 0}$ that satisfy the binary marginals property are \emph{monotone}, that is, $f(S + a) \geq f(S)$ for every $S \subseteq [m]$ and  $a \in [m]$.

The agents' costs $c_i: 2^{[m]} \mapsto \mathbb{R}_{\geq 0}$ can be classified into a variety of classes of set functions following the hierarchy of complement-free (or substitute-free) functions \cite{AGTBook}. The current work focuses on one such class that captures complementarity among chores. In particular, we study chore division settings wherein, for each agent $i \in [n]$, the cost function $c_i$ (in addition to satisfying the binary marginals property) is supermodular, that is, $c_i(S+a) - c_i(S) \leq c_i(T + a) - c_i(T)$ for all subsets $S \subseteq T$ and chores $a \in [m] \setminus T$.\footnote{Supermodular cost functions are a subclass of superadditive functions since they satisfy $c_i(S) + c_i(T) \leq c_i(S \cup T)$ for every disjoint pair of subsets $S, T \subseteq [m]$.} 


Throughout this work, we will assume that the cost functions of agents can be accessed through oracles that answer value queries: given any subset $S \subseteq [m]$, the oracle for agent $i$ returns $c_i(S)$, the cost of bundle $S$ for agent $i$. In particular, our algorithmic results hold in the value-oracle model and do not require explicit descriptions of the cost functions. \\

\noindent
{\bf Economic efficiency.} An allocation $\Alloc = (A_1, \ldots, A_n)$ is said to \emph{Pareto dominate} another allocation $\AllocB = (B_1,\ldots, B_n)$ iff for each agent $i \in [n]$, we have $c_i(A_i) \leq c_i(B_i)$ and for some agent $j \in [n]$ the inequality is strict (i.e., $c_j(A_j) < c_j(B_j)$). An allocation $\Alloc$ is said to be \emph{Pareto efficient} iff there is no other complete allocation that Pareto dominates it.

For an (partial) allocation $\Alloc = (A_1, A_2, \ldots, A_n)$, the utilitarian social cost $\SC(\Alloc)$ is the sum of the costs of the agents, $\SC(\Alloc) \coloneqq  \sum_{i=1}^{n} c_i(A_i)$; we will simply refer to $\SC(\Alloc)$ as the social cost of $\Alloc$. An allocation $\Alloc^* \in \Pi_n([m])$ that minimizes the social cost among the set of all complete allocations is called a \emph{social cost minimizing} allocation, and we write $c^* \coloneqq \SC(\Alloc^*)$ to denote the \emph{minimum social cost} of the instance. It follows that every social cost minimizing allocation is also Pareto efficient. \\

\noindent
{\bf Fairness.} An (partial) allocation $\Alloc = (A_1, \ldots, A_n)$ is said to be \emph{envy-free up to one chore} ($\EFone$) iff for every pair of agents $i,j \in [n]$ with $A_i \neq \emptyset$, \emph{there exists} some chore $t \in A_i$ such that $c_i(A_i \setminus \{t\}) \leq c_i(A_j)$. That is, agent $i$'s envy for agent $j$ vanishes upon the removal of some chore from her own bundle. $\Alloc = (A_1, \ldots, A_n)$ is said to be \emph{envy-free up to any chore} ($\EFX$) iff for every pair of agents $i,j \in [n]$ with $A_i \neq \emptyset$, and \emph{every} chore $t \in A_i$, we have $c_i(A_i \setminus \{t\}) \leq c_i(A_j)$. That is, agent $i$'s envy for agent $j$ vanishes upon the removal of any chore from her own bundle. 

Additionally, we consider the \emph{minimax share} ($\MMS$) fairness guarantee, which is based on an interpretation of the well-known cut-and-choose protocol. An allocation $\Alloc = (A_1, \ldots, A_n)$ in a fair division instance $\langle [n], [m], \{c_i\}_{i \in [n]} \rangle$ is said to be $\MMS$-fair iff, for each agent $i \in [n]$, the cost $c_i(A_i) \leq \tau_i$. Here, $\tau_i$ is called the {minimax share} of agent $i$ and is defined as follows

\begin{align}
\tau_i \coloneqq \min\limits_{\substack{(X_1, \ldots, X_n) \in \Pi_n([m])}} \ \   \max\limits_{j \in [n]} \ c_i(X_j) \label{eq:defn-minimax}
\end{align}

The minimax share $\tau_i$ can be interpreted as follows: suppose agent $i$ gets to divide all the chores $[m]$ into $n$ bundles, and is then assigned the bundle with the highest cost (according to $c_i$). In order to minimize her cost, agent $i$ will divide the chores $[m]$ into bundles $(X_1, X_2, \ldots, X_n)$ such that the maximum cost among the bundles (with respect to $c_i$) is minimized. The minimax share $\tau_i$ of agent $i$ is the minimum cost of the bundle that agent $i$ can obtain in this thought experiment.

For every allocation $\Alloc = (A_1,\ldots, A_n)$ we define the \emph{cost profile} of $\Alloc$ to be the $n$-tuple of the costs of the agents, $(c_1(A_1), c_2(A_2),\allowbreak \ldots, c_n(A_n))$. Furthermore, for any allocation $\Alloc = (A_1,\ldots, A_n)$, the sorted cost profile $\sigma(\Alloc)$ is the $n$-tuple obtained by sorting the profile of costs incurred by the agents in nonincreasing order. That is, $\sigma(\Alloc) = (c_{i_1}(A_{i_1}), c_{i_2}(A_{i_2}) \ldots, c_{i_n}(A_{i_n}))$, where $c_{i_1}(A_{i_1}) \geq c_{i_2}(A_{i_2}) \ldots \geq c_{i_n}(A_{i_n})$ and $i_1, i_2, \ldots, i_n$ is a permutation of $[n]$.

Based on this, we now define \emph{Lorenz domination} and obtain the partial order induced by this notion over the allocations. Formally, allocation $\Alloc = (A_1, \ldots, A_n)$ \emph{Lorenz dominates} allocation $\AllocB = (B_1,\ldots, B_n)$ iff, for each index $k \in [n]$, the sum of the first $k$ components of $\sigma(\Alloc)$ is at most the sum of the first $k$ components of $\sigma(\AllocB)$. That is, $\sum_{i=1}^k \sigma(\Alloc)_i \leq \sum_{i=1}^k \sigma(\AllocB)_i$, where, $\sigma(\Alloc)_i$ and $\sigma(\AllocB)_i$ denote the $i^{th}$ component of $\sigma(\Alloc)$ and $\sigma(\AllocB)$, respectively. We write $\sigma(\Alloc) \geq_{L} \sigma(\AllocB)$ to denote that allocation $\Alloc$ Lorentz dominates allocation $\AllocB$. An allocation $\Alloc$ is said to be \emph{Lorenz dominating} iff it Lorenz dominates all other complete allocations. We note that, in general, a Lorenz dominating allocation may not exist, but if it does then it is also social-cost minimizing (and, hence, Pareto efficient) and also \emph{leximin}.\footnote{An allocation $\Alloc$ is called leximin iff it lexicographically minimizes the sorted cost profile $\sigma(\Alloc)$ among the set of all complete allocations.} \\

\noindent
{\bf Matroids and rank functions.}
A matroid is defined as a tuple $\mat = ([m], \mathcal{I})$ where $[m]$ is the ground set and $\mathcal{I} \subseteq 2^{[m]}$ is the set of all independent sets. The subsets $A \in \mathcal{I}$ are called the independent sets of the matroid $\mathcal{M}$. A matroid satisfies the following two properties: $(i)$ the hereditary property, which mandates that if subset $A \in \mathcal{I}$, then every subset $B \subseteq A$ must also be independent, i.e. $B \in \mathcal{I}$, and $(ii)$ the augmentation property, which requires that for every pair of independent sets $A,B \in \mathcal{I}$ satisfying $|B| < |A|$, there must exist an element $a \in A \setminus B$ such that $B + a \in \mathcal{I}$. Any cardinality-wise largest independent set of a matroid $\mat$ is called a \emph{basis} of $\mat$.

Associated with every matroid $\mat = ([m], \mathcal{I})$, there is a rank function $r: 2^{[m]} \mapsto \mathbb{R}_{\geq 0}$ that specifies for any subset $S \subseteq [m]$ the size of largest independent subset within $S$; formally, $r(S) \coloneqq \max\{|I| : I \subseteq S \text{ and } I \in \mathcal{I}\}$. Note that, by definition, matroid-rank functions are monotone and they satisfy $0 \leq r(S) \leq |S|$ for all subsets $S \subseteq [m]$. The equality $r(S) = |S|$ holds iff $S \in \mathcal{I}$. The following characterization is well-known: a set function $r: 2^{[m]} \mapsto \mathbb{R}_{\geq 0}$ is a matroid-rank function iff $r$ is a binary submodular function~\cite{schrijver2003combinatorial}. Since every matroid-rank function $r$ has binary marginals, for any given subset $A \subseteq [m]$, we can efficiently compute a subset $I \subseteq A$ that satisfies $r(A) = r(I) = |I|$, i.e., $I$ is a largest-cardinality independent subset of $A$. 

Given matroids $\mat_1 = ([m], \mathcal{I}_1), \ldots, \mat_n = ([m], \mathcal{I}_n)$, we can define a new matroid by considering their \emph{union}~\cite{schrijver2003combinatorial}. Such a union is denoted as $\umat = \cup_{i \in [n]} \mat_i = ([m], \widehat{\mathcal{I}})$ and its independent sets are composed as follows $\widehat{\mathcal{I}} \coloneqq \{A_1 \cup \ldots \cup A_n : A_i \in \mathcal{I}_i \text{ for all } i \in [n]\}$. Equivalently, $S \subseteq [m]$ is an independent set of $\umat$ iff $S$ can be partitioned into subsets $S_1, S_2, \ldots, S_n$ with the property that $S_i \in \mathcal{I}_i$, for all $i \in [n]$. One can verify that $\umat$ satisfies the hereditary and augmentation properties and, hence, is a matroid~\cite{schrijver2003combinatorial}. The rank function of $\umat$ is denoted by $\widehat{r}$.

The well-known matroid union theorem \cite[Corollary~42.1a]{schrijver2003combinatorial} gives us an expression for $\widehat{r}$ in terms of the rank functions $r_1, \ldots, r_n$ of the matroids $\mat_1, \ldots, \mat_n$, respectively. The matroid union theorem is stated next. 
\begin{theorem}
\label{lemma:mat-union}
If $\umat$ is the union of matroids $\mat_1 = ([m], \mathcal{I}_1)$, $\mat_2 = ([m], \mathcal{I}_2)$, $\ldots$, $\mat_n = ([m], \mathcal{I}_n)$ with rank functions $r_1, \dots r_n: 2^{[m]} \mapsto \mathbb{R}_{\geq 0}$ respectively, then, the rank function $\widehat{r}$ of $\umat$ can be expressed as 
$ \widehat{r}(S) \coloneqq \min_{T \subseteq S} \Bigg( |S \setminus T| + \sum_{i \in [n]} r_i(T) \Bigg)$ for all $S \subseteq [m]$.
\end{theorem}

For ease of notation, we will use $\mat_{i \times k}$ to denote the $k$-fold union of matroid $\mat_i$, i.e., $\mathcal{M}_{i \times k} = \cup_{j \in [k]} \mat_i$ for any integer $k \in \mathbb{Z}_+$. Additionally, we will use $r_{i \times k}$ to denote the rank function of $\mat_{i \times k}$.

\section{Social Cost Minimization} \label{sec:cost-min}
We begin by studying the problem of minimizing the utilitarian social cost for instances $\langle [n], [m], \{c_i\}_{i \in [n]} \rangle$ with binary supermodular cost functions. The main result of this section is a polynomial-time algorithm (Algorithm \ref{algo:social-cost-min}) for this problem. The algorithm only requires  value-oracle access to the cost functions. 

First, we develop a characterization of the minimum social cost for instances with binary supermodular cost functions (Lemma \ref{lemma:min-cost-characterization}). Specifically, for such instances, we provide an expression for the minimum social cost; this expression is crucially used in proofs throughout the subsequent sections. Here, the key idea is based on the fact that for every binary supermodular cost function $c_i$ we can associate a matroid $\mat_i$ and, then, express the minimum social cost (under cost functions $\{c_i\}_{i \in [n]}$) in terms of the rank of the union matroid $\umat = \cup_{i \in [n]} \mat_i$. 

This characterization enables us to efficiently construct (via the matroid-union algorithm) a {\it partial} allocation of chores $\AllocB = (B_1, \ldots, B_n)$ with the following useful properties: The partial allocation $\AllocB$ has zero social cost, and \emph{every} possible way to complete the partial allocation $\AllocB$ (by assigning the unallocated chores $[m] \setminus (\cup_{i \in [n]} B_i)$ among the agents) results in a social cost minimizing complete allocation. Algorithm \ref{algo:social-cost-min} first computes such a partial allocation (via Subroutine \ref{subroutine:cost-min-partial-alloc}) and then allocates the remaining chores arbitrarily, thereby obtaining a social cost minimizing allocation.

The following lemma shows that there is a bijection between binary supermodular costs and matroid-rank functions. The proof of the following lemma appears in Appendix~\ref{appendix:missing-proofs-sc}.

\begin{restatable}{lemma}{SubmodSupermod}
\label{lemma:submod-supermod}
A set function $f:2^{[m]} \mapsto \mathbb{R}_{\geq 0}$ is binary supermodular iff the set function $g:2^{[m]} \mapsto \mathbb{R}_{\geq 0}$, defined as follows, is a matroid-rank function:  $g(S) \coloneqq |S| - f(S)$, for all $S \subseteq [m]$.
\end{restatable}


Lemma \ref{lemma:submod-supermod} shows that for each agent $i \in [n]$, with binary supermodular cost function $c_i$, we can obtain a matroid-rank function $r_i : 2^{[m]} \mapsto \mathbb{R}_{\geq 0}$ by defining $r_i(S) = |S| - c_i(S)$, for all subsets $S \subseteq [m]$. Therefore, we can identify a matroid $\mat_i$ whose rank function is $r_i$. The lemma below builds upon this to provide an expression for the minimum social cost in terms of the rank of the union matroid $\umat = \cup_{i \in [n]} \mat_i$. 

\begin{restatable}{lemma}{mincostchar}\label{lemma:min-cost-characterization}
Consider a fair chore division instance $\langle [n], [m], \{c_i\}_{i \in [n]}\rangle$ with binary supermodular costs $\{c_i\}_i$. For each agent $i \in [n]$, let $r_i$ be the matroid-rank function obtained by setting $r_i(S) = |S| - c_i(S)$, for all subsets $S \subseteq [m]$. Also, let $\mat_i = ([m], \mathcal{I}_i)$ denote the matroid whose rank function is $r_i$ and $\widehat{r}$ denote the rank-function of the union matroid $\umat = \cup_{i=1}^n \mat_i$. Then, the minimum social cost $c^*$ of the instance $\langle [n], [m], \{c_i\}_{i \in [n]}\rangle$ satisfies $c^* = m - \widehat{r}([m])$.
\end{restatable}

\begin{proof}
By definition, for the minimum social cost $c^*$ we have
\begin{align*}
c^* & = \min_{(A_1, \ldots, A_n) \in \Pi_n([m])} \ \sum_{i \in [n]} c_i(A_i) \\
& = \min_{(A_1, \ldots, A_n) \in \Pi_n([m])} \ \sum_{i \in [n]} \big( |A_i| - r_i(A_i) \big) \tag{since $r_i(A_i) = |A_i| - c_i(A_i)$} \\
& = m - \max_{(A_1, \ldots, A_n) \in \Pi_n([m])} \sum_{i \in [n]} r_i(A_i) \numberthis \label{lemma-min-cost:eq1}
\end{align*}
The last equality uses the fact that $m = \sum_{i \in [n]} |A_i|$; recall that we are considering the social cost of complete allocations $(A_1, \allowbreak \ldots, A_n)$. To complete the proof we will show that 

\begin{align}
\label{lemma-min-cost:eq2}
\max_{(A_1, \ldots, A_n) \in \Pi_n([m])} \sum_{i \in [n]} r_i(A_i) = \widehat{r}([m]).
\end{align}
Indeed, equations (\ref{lemma-min-cost:eq1}) and (\ref{lemma-min-cost:eq2}) together imply the desired equality $c^* = m - \widehat{r}([m])$. 

To establish equation (\ref{lemma-min-cost:eq2}), consider any basis $B$ of $\umat$, which (by definition of the matroid union) can be partitioned into subsets $B_1, B_2, \ldots, B_n$ with the property that $B_i \in \mathcal{I}_i$, for all $i \in [n]$. Since $B$ is a basis of $\umat$, we have 
\begin{align}
\widehat{r}([m]) &= |B| = \sum_{i \in [n]} |B_i| = \sum_{i \in [n]} r_i(B_i) &\nonumber\\ &\leq \max_{(A_1, \ldots, A_n) \in \Pi_n([m])} \sum_{i \in [n]} r_i(A_i) \label{ineq:rank-up}
\end{align}
Here, the last inequality follows from the fact that the rank functions $r_i$ are monotonic. In particular, by arbitrarily assigning the chores remaining in $[m] \setminus \cup_{i \in [n]} B_i$, we can extend the partial allocation $(B_1, \ldots, B_n)$ into a complete allocation $(B'_1, \ldots, B'_n)$. We then obtain $\sum_{i \in [n]} r_i(B_i) \leq \sum_{i \in [n]} r_i(B'_i) \leq \max_{(A_1, \ldots, A_n) \in \Pi_n([m])} \allowbreak \sum_{i \in [n]} r_i(A_i)$.
 
Next, we show that the inequality (\ref{ineq:rank-up}) cannot be strict. Assume, towards a contradiction, that exists a complete allocation $(A_1, \ldots, A_n)$ such that $\widehat{r}([m]) <  \sum_{i \in [n]} r_i(A_i)$. Note that within each subset $A_i$ there exists an independent set $A'_i \in \mathcal{I}_i$ such that $r_i(A_i) = r_i(A'_i) = |A'_i|$. Therefore, $\sum_{i \in [n]} r_i(A_i) = \sum_{i \in [n]} |A'_i|$. Furthermore, the fact that $A'_i \in \mathcal{I}_i$, for each $i \in [n]$, implies that the union $\cup_{i \in [n]} \widehat{A}'_i$ is an independent set of $\umat$. Hence, $\widehat{r}([m]) \geq |\cup_{i \in [n]} \widehat{A}'_i| = \sum_{i \in [n]} |A'_i|$. Hence, by contradiction, we obtain $\widehat{r}([m]) \geq \max_{(A_1, \ldots, A_n) \allowbreak \in \Pi_n([m])} \allowbreak \sum_{i \in [n]} r_i(A_i)$. This inequality and (\ref{ineq:rank-up}) imply equation (\ref{lemma-min-cost:eq2}). 

The lemma stands proved. 
\end{proof}

\floatname{algorithm}{Subroutine}
\begin{algorithm}[ht]
\caption{\textsc{CostMinPartAlloc}}
\begin{tabularx}{\textwidth}{X}
{\bf Input:} Chore division instance  $\langle [n], [m], \{c_i\}_{i \in [n]}\rangle$ with value-oracle access to the costs. \\
{\bf Output:} Partial allocation $\AllocB$.
\end{tabularx}
  \begin{algorithmic}[1]
  		\STATE Define for each $i \in [n]$, matroid-rank functions $r_i(S) \coloneqq |S| - c_i(S)$ for all $S \subseteq [m]$ and let $\mathcal{M}_i = ([m], \mathcal{I}_i)$ be the matroid corresponding to the function $r_i$. \label{line:subroutine-mrf}
  		\STATE Let union matroid $\widehat{\mathcal{M}} \coloneqq  \cup_{i \in [n]} \mathcal{M}_i$. \label{line:subroutine-union-mat}
  		\STATE Compute a basis $B$ of the matroid $\widehat{\mathcal{M}}$ along with an $n$-partition of $B$ into subsets $(B_1, B_2, \ldots, B_n)$ such that $B_i \in \mathcal{I}_i$ via the matroid union algorithm \cite{schrijver2003combinatorial}. \label{line:subroutine-matroid-union-alg}
  		\STATE \textbf{return } $\AllocB = \left( B_1, B_2, \ldots, B_n \right)$.
		\end{algorithmic}
		\label{subroutine:cost-min-partial-alloc}
\end{algorithm}

The above lemma forms the basis of Subroutine \ref{subroutine:cost-min-partial-alloc} which is the key component of Algorithm \ref{algo:social-cost-min}. For each agent $i \in [n]$, let $\mat_i = ([m], \mathcal{I}_i)$ be the matroid whose rank-function is defined as $r_i(S) = |S| - c_i(S)$ and $\umat = \cup_{i \in [n]} \mat_i$ be the union of matroids $\{\mat_i\}_{i \in [n]}$. Subroutine \ref{subroutine:cost-min-partial-alloc} uses the matroid union algorithm \cite[Chapter 42.3]{schrijver2003combinatorial} to compute a basis $B$ of the union matroid $\umat$ and the corresponding $n$-partition of $B$ into $(B_1, B_2, \ldots, B_n)$. This $n$-partition, which is also a partial allocation, is then returned by the subroutine. 

In the following lemma we show that \emph{every} possible way of completing the partial allocation $(B_1,\ldots, B_n)$ by assigning the unallocated chores ($[m] \setminus \cup_{i \in [n]} B_i$) among the agents results in a social cost minimizing allocation. This property is crucially used in computing fair and efficient allocations in the subsequent sections.


\begin{lemma}
\label{lemma:completing-the-partial-alloc} 
For a fair chore division instance $\langle [n], [m], \{c_i\}_{i \in [n]}\rangle$, let $\mathcal{B} = (B_1,\ldots, B_n)$ be the partial allocation returned by Subroutine \ref{subroutine:cost-min-partial-alloc}. Then, any complete allocation $\overline{\AllocB} = (\overline{B}_1, \ldots, \overline{B}_n) \in \Pi_n([m])$ with the property that $B_i \subseteq \overline{B}_i$, for all $i \in [n]$, satisfies $\sum_{i \in [n]} c_i(\overline{B}_i) = c^*$, i.e., $\overline{\AllocB}$ is a social cost minimizing complete allocation in the given instance. 
\end{lemma}
\begin{proof}
For the returned partial allocation $\mathcal{B} = (B_1,\ldots, B_n)$, write $U$ to denote the set of unallocated chores, $U \coloneqq [m] \setminus \big( \cup_{i \in [n]} B_i \big)$. To establish the lemma we will prove that \\
\noindent
({\rm I}) The cost $c_i(B_i) = 0$, for all agents $i \in [n]$, and \\
\noindent
({\rm II}) The number of unassigned chores $|U| = c^*$.
 
Note that, together, ({\rm I}) and ({\rm II}) imply the lemma: starting from the partial allocation $\mathcal{B} = (B_1, B_2, \ldots, B_n)$, we allocate all the chores in $U$ and obtain complete allocation $\overline{\AllocB} = (\overline{B}_1, \ldots, \overline{B}_n)$. Therefore, $\sum_{i=1}^n c_i(\overline{B}_i) \leq \sum_{i=1}^n c_i(B_i) + |U| = c^*$; these bounds follow from ({\rm I}), ({\rm II}), and the fact that $c_i$s have binary marginals.

Furthermore, by definition of $c^*$ we must have $\sum_{i \in [n]} c_i(\overline{B}_i) \geq  c^*$. Therefore, the complete allocation $\overline{\AllocB}$ must be a social cost minimizing allocation, $\sum_{i \in [n]} c_i(\overline{B}_i) = c^*$.

For establishing ({\rm I}) we note that $c_i(B_i) = |B_i| - r_i(B_i) = |B_i| - |B_i| = 0$, for all agents $i \in [n]$. Here, the first equality follows from Line \ref{line:subroutine-mrf} and second follows from Line \ref{line:subroutine-matroid-union-alg}; recall that $B_i$ is an independent set of matroid $\mathcal{M}_i$ and, hence, $r_i(B_i) = |B_i|$. 

For proving ({\rm II}) we use the fact that $|U| = |[m] \setminus \big( \cup_{i \in [n]} B_i \big)| = m - |\cup_{i \in [n]} B_i| = m - \widehat{r}([m])$; the last equality follows since $\cup_{i \in [n]} B_i$ is a basis of $\umat$ (Line \ref{line:subroutine-matroid-union-alg}). Therefore, Lemma \ref{lemma:min-cost-characterization} (i.e., the equality $m - \widehat{r}([m]) = c^*$) gives us the desired bound $|U| = m - \widehat{r}([m]) = c^*$.

As mentioned previously, ({\rm I}) and ({\rm II}) imply the lemma. This completes the proof. 
\end{proof}

\floatname{algorithm}{Algorithm}
\begin{algorithm}[ht]
\caption{\textsc{SocialCostMin}} \label{algo:social-cost-min}
\begin{tabularx}{\textwidth}{X}
{\bf Input:} Chore division instance  $\langle [n], [m], \{c_i\}_{i \in [n]}\rangle$ with value-oracle access to the costs $\{c_i\}_i$. \\
{\bf Output:} Social-cost minimizing complete allocation $\Alloc$. 
\end{tabularx}
 \begin{algorithmic}[1]
\STATE Set $( A_1,\ldots, A_n) = \textsc{CostMinPartAlloc}(\langle [n], [m], \{c_i\}_{i} \rangle)$ and define set of unallocated chores $U \coloneqq  [m] \setminus \left( \cup_{i \in [n]} A_i \right)$.
\WHILE{$U \neq \emptyset$} \label{alg1:loop-begin}
\STATE Select an arbitrary chore $t \in U$ and an arbitrary agent $i \in [n]$. Update $A_i \gets A_i + t$ and $U \gets U - t$. 
\ENDWHILE \label{alg1:loop-end}
\RETURN $\Alloc = \left( A_1,\ldots, A_n \right)$.
\end{algorithmic}
\end{algorithm}

Our social cost minimizing algorithm (Algorithm \ref{algo:social-cost-min}) first uses Subroutine \ref{subroutine:cost-min-partial-alloc} to obtain the aforementioned partial allocation, and then assigns the unallocated chores arbitrarily.  The social-cost guarantee achieved by Algorithm \ref{algo:social-cost-min} is formalized in the theorem below. 

\begin{restatable}{theorem}{ThmSocialCostMin}
\label{theorem:social-cost-min}
Given any chore division instance  $\langle [n], [m], \{c_i\}_{i \in [n]}\rangle$ with value-oracle access to the binary supermodular cost functions $\{c_i\}_{i \in [n]}$, Algorithm \ref{algo:social-cost-min} computes a social-cost-minimizing, complete allocation in polynomial time.
\end{restatable}
\begin{proof}
The theorem directly follows from Lemma \ref{lemma:completing-the-partial-alloc} and the fact that Subroutine \ref{subroutine:cost-min-partial-alloc} can be executed in polynomial time in the value-oracle model. In particular, it is known that the matroid union algorithm (in Line \ref{line:subroutine-matroid-union-alg} of Subroutine \ref{subroutine:cost-min-partial-alloc}) only requires value-oracle access to the rank functions $r_i$s, i.e., only requires the values of $r_i(S)$ upon querying for subsets $S \subseteq [m]$ (see, e.g.,~\cite{schrijver2003combinatorial}). One can directly answer such value queries, given a value oracle for the cost functions $c_i$s; specifically, for any queried subset $S \subseteq [m]$, we can first query for the cost $c_i(S)$ and then return $|S| - c_i(S)$. Hence, we can execute the subroutine with value-oracle access to the cost functions $c_i$s. 

These observations imply that Algorithm \ref{algo:social-cost-min} runs in polynomial time and finds a complete allocation that is social-cost minimizing. 
\end{proof}

\section{Fair and Efficient Chore Allocation} \label{sec:positive} 
This section establishes the existence and efficient computability of fair and efficient allocations of chores under binary supermodular costs. We focus on three standard notions of fairness. Specifically, we show the existence of 
\begin{itemize}
\item Complete allocations that are $\EFone$ and social-cost minimizing (among all complete allocations), 
\item Complete allocations that are $\MMS$-fair and social-cost minimizing (again, among all complete allocations), and 
\item Complete allocations that are Lorenz dominating (across all complete allocations).
\end{itemize}
Furthermore, we develop polynomial-time algorithms (in the value oracle model) for finding such fair and efficient allocations. In each of these three cases, we start with the partial allocation computed by Subroutine \ref{subroutine:cost-min-partial-alloc} and then assign the unallocated chores in a case-specific manner to obtain the desired fairness guarantee. 

\subsection{$\EFone$ and Economic Efficiency}
For chore division instances with binary supermodular costs, we develop a polynomial-time algorithm (Algorithm \ref{algo:PO-EF1}) for finding complete allocations that are $\EFone$ and social-cost minimizing (among all complete allocations). The fact that our algorithm necessarily succeeds in finding such an allocation proves the existence of $\EFone$ and Pareto efficient chore allocations in the  context of binary supermodular costs.\footnote{Recall that a social-cost minimizing allocation is necessarily Pareto efficient.} 

 Algorithm \ref{algo:PO-EF1} first invokes Subroutine \ref{subroutine:cost-min-partial-alloc} to compute a partial allocation $\Alloc=(A_1, \ldots, A_n)$ of chores. Then, the remaining chores, $U = [m] \setminus \left( \cup_{i \in [n]} A_i \right)$, are assigned to extend $\Alloc$ into a complete $\EFone$ allocation. Lemma \ref{lemma:completing-the-partial-alloc} ensures that---irrespective of how we extend the allocation---the complete allocation obtained at the end is social-cost minimizing. Hence, the key concern here is to obtain the $\EFone$ guarantee. Towards this, we show (in the proof of Theorem \ref{thm:po-ef1}) that, as we iteratively assign the chores from $U$, for each maintained partial allocation there always exists an agent $\ell \in [n]$ who does not envy any other agent. Updating the partial allocation by assigning the next chore to $\ell$ maintains $\EFone$. Therefore, the complete allocation returned by the algorithm is also $\EFone$. This guarantee is formalized in the theorem below.
 
\floatname{algorithm}{Algorithm}
\begin{algorithm}[ht]
\caption{\textsc{EF1andEfficient}} \label{algo:PO-EF1}
\begin{tabularx}{\textwidth}{X}
{\bf Input:} Chore division instance  $\langle [n], [m], \{c_i\}_{i \in [n]}\rangle$ with value-oracle access to the costs $\{c_i\}_i$. \\
{\bf Output:} Complete allocation $\Alloc = (A_1, \ldots, A_n)$.
\end{tabularx}
  \begin{algorithmic}[1]
  		\STATE Set $\left( A_1,\ldots, A_n \right) = \textsc{CostMinPartAlloc}(\langle [n], [m], \{c_i\}_{i} \rangle)$.  \label{algPO+EF1:subroutine-call}
		\STATE Define set of unallocated chores $U \coloneqq  [m] \setminus \big( \cup_{i \in [n]} A_i \big)$.
  		\WHILE{$U \neq \emptyset$} \label{algPO+EF1:loop-begin}
			\STATE Let $\ell \in [n]$ be an agent with the property that $c_\ell(A_\ell) \leq c_\ell(A_j)$ for all $j \in [n]$. \label{algPO+EF1:no-envy-agent} \\
			\COMMENT{We will prove that such an agent necessarily exists.}
			\STATE Select any chore $t \in U$ and update $A_\ell \gets A_\ell + t$ along with $U \gets U - t$. \label{algPO+EF1:chore-assign}
			\ENDWHILE \label{algPO+EF1:loop-end}
		\RETURN $\Alloc = \left( A_1, A_2, \ldots, A_n \right)$
		\end{algorithmic}
\end{algorithm}

\begin{restatable}{theorem}{poefone}\label{thm:po-ef1}
For binary supermodular cost functions, complete allocations that are $\EFone$ and social-cost minimizing (among all complete allocations) always exist and can be computed in polynomial time via Algorithm \ref{algo:PO-EF1} (given value-oracle access to the costs).
\end{restatable}
\begin{proof}
We first establish that the returned allocation is $\EFone$. Note that the initial partial allocation computed by Subroutine \ref{subroutine:cost-min-partial-alloc} (in Line \ref{algPO+EF1:subroutine-call} of Algorithm \ref{algo:PO-EF1}) is $\EFone$. In fact, it is envy-free, since for all agents $i \in [n]$, we have $c_i(A_i) = |A_i| - r_i(A_i) = |A_i| - |A_i| = 0$; recall that the definition of the rank function $r_i$ and the fact that for the bundles $A_i$ returned by Subroutine \ref{subroutine:cost-min-partial-alloc} satisfy $A_i \in \mathcal{I}_i$. 

We note that for every partial allocation considered in the while loop of the algorithm, there exists an agent $\ell \in [n]$ who does not envy any other agent:\footnote{That is, Line \ref{algPO+EF1:no-envy-agent} of the algorithm successfully identifies an agent.} Assume, towards a contradiction, that for a maintained partial allocation there does not exist such an agent $\ell$. In such a case, we would have an envy-cycle among the agents. In particular, there would exist, for $k \geq 2$, a set of agents $a_1, a_2, \ldots, a_k$ such that $c_{a_i}(A_{a_i}) > c_{a_{i}}(A_{a_{i+1}})$, for all $i \in [k]$. Here,  index $k+1$ is cyclically mapped to $1$. Resolving this envy-cycle (i.e., transferring bundles $A_{a_{i+1}}$ to agent $a_i$ for all $i \in [k]$) will reduce the social cost of the partial allocation. That is, the social cost of the returned allocation will end up being lower than $c^*$, the lowest possible social cost. This leads to a contradiction. Hence, in Line \ref{algPO+EF1:no-envy-agent}, we 
can always find an agent $\ell \in [n]$ who does not envy anyone else. 

We can assign an arbitrary chore $t \in U$ to such an agent $\ell \in [n]$ (in Line \ref{algPO+EF1:chore-assign}) and maintain $\EFone$. This follows from the observation that, for any other agent $j \in [n] \setminus \{\ell\}$, the cost of agent $\ell$'s bundle, $c_j(A_\ell)$, does not decrease and $c_j(A_j)$ remains unchanged. Hence, $\EFone$ is preserved for all agents $j \neq \ell$. Now, for agent $\ell$ one can remove the newly assigned chore $t$ from her bundle to return to envy freeness from $\ell$'s perspective, i.e., even after the inclusion of $t$ we have $c_\ell(A_\ell - t) \leq c_\ell(A_j)$. Therefore, the (parital) allocation continues to be $\EFone$ after the assignment. Using this invariant we obtain that the returned complete allocation is also $\EFone$. 

The economic efficiency of the complete allocation returned by Algorithm \ref{algo:PO-EF1} follows from Lemma \ref{lemma:completing-the-partial-alloc}. In particular, the lemma ensures that returned allocation is social-cost minimizing and, hence, Pareto efficient. Finally, note that the while loop and Subroutine \ref{subroutine:cost-min-partial-alloc} can be executed in polynomial time, given value-oracle access to the cost functions $\{c_i\}_i$. Hence, Algorithm \ref{algo:PO-EF1} runs in polynomial time. The theorem stands proved. 
\end{proof}

\subsection{$\MMS$-Fairness and Economic Efficiency}
\label{section:mms-po}

Considering chore division instances with binary supermodular costs, this section develops a polynomial-time algorithm (Algorithm \ref{algo:PO-MMS}) for finding complete allocations that are $\MMS$-fair and social-cost minimizing. Our algorithm executes in the value-oracle model. Also, the fact that Algorithm \ref{algo:PO-MMS} necessarily succeeds proves the existence of $\MMS$-fair and PO allocations in the current context.

Towards finding $\MMS$-fair allocations under binary supermodular costs, we first provide an expression for the minimax shares $\tau_i$ for the agents $i \in [n]$. We show that here $\tau_i$s can be computed efficiently using the matroid union algorithm (Lemma \ref{lemma:minimax-share}). In addition, as part of our key technical lemma (Lemma \ref{lemma:sum-of-shares}), we show that (in the context of binary supermodular costs) the sum of the minimax shares of the agents $i \in [n]$ is always at least the minimum social cost, $c^*$, of the instance, i.e., $\sum_{i \in [n]} \tau_i \geq c^*$. We will establish this inequality using the matroid union theorem (Theorem \ref{lemma:mat-union}). 

Note that the inequality $\sum_{i \in [n]} \tau_i \geq c^*$ must be satisfied for the existence of any $\MMS$-fair allocation. Otherwise, for any allocation $\Alloc= (A_1, \ldots, A_n)$ we would have $\sum_{i\in [n]} c_i(A_i) \geq c^* > \sum_{i \in [n]} \tau_i$. Hence, for some agent $i$ it must hold that $c_i(A_i) > \tau_i$, which would violate the $\MMS$ requirement.  

Algorithm \ref{algo:PO-MMS} first invokes Subroutine \ref{subroutine:cost-min-partial-alloc} to compute a partial allocation of chores and then completes it, thereby computing a social-cost-minimizing allocation (Lemma \ref{lemma:completing-the-partial-alloc}). The inequality $\sum_{i \in [n]} \tau_i \geq c^*$ is then used to argue that the unallocated chores $U = [m] \setminus \left( \cup_{i \in [n]} A_i \right)$ can always be assigned in a way that the cost of each agent in the resulting allocation is at most its minimax share, i.e., the final allocation satisfies the $\MMS$ guarantee.

We begin by deriving an expression for the minimax shares $\tau_i$s (see equation (\ref{eq:defn-minimax})) and proving that they can be computed efficiently. Note that the rank function $r_i$ associated with the cost $c_i$ is satisfies $r_i(S) = |S| - c_i(S)$, for all subsets $S \subseteq [m]$. Also, recall that $\mat_i$ denotes the matroid associated with the rank function $r_i$ and $\mat_{i \times n}$ denotes the $n$-fold union of matroid $\mat_i$, i.e., $\mathcal{M}_{i \times n} = \cup_{j \in [n]} \mat_i$. Additionally, $r_{i \times n}$ denotes the rank function of $\mat_{i \times n}$. The proof of the following lemma is deferred to  Appendix~\ref{appendix:proofs-positive}.

\begin{restatable}{lemma}{minimaxshare}\label{lemma:minimax-share}
For any chore division instance $\langle [n], [m], \{c_i\}_{i \in [n]}\rangle$ with binary supermodular costs, the minimax share $\tau_i$, of each agent $i \in [n]$, satisfies $\tau_i = \left\lceil \frac{m - r_{i \times n}([m]) }{n} \right\rceil$.
In addition, given value-oracle access to the costs, the shares $\tau_i$s can be computed in polynomial time. 
\end{restatable}

The following lemma provides a key technical bound of this section asserting that the minimum social cost, $c^*$, is at most the sum of the minimax share $\tau_i$s.  In Theorem \ref{thm:po-mms} we make use of this inequality to show that $\MMS$ allocations always exist under binary supermodular valuations. 


\begin{restatable}{lemma}{sumofshares}\label{lemma:sum-of-shares}
For any chore division instance $\langle [n], [m], \{c_i\}_{i}\rangle$, with binary supermodular costs, let $\tau_i$ denote the minimax share of each agent $i \in [n]$ and $c^*$ be the minimum social cost. Then, we have $c^* \leq \sum_{i \in [n]} \tau_i$.
\end{restatable}
\begin{proof}
Consider agent $i \in [n]$. From Lemma \ref{lemma:minimax-share}, we know that 

\begin{align*}
\tau_i = \left\lceil \frac{m - r_{i \times n}([m]) }{n} \right\rceil \geq \frac{m - r_{i \times n}([m]) }{n}.
\end{align*}
Using matroid union theorem (Theorem \ref{lemma:mat-union}) we can express the rank function $r_{i \times n}([m])$ (of the $n$-fold union matroid) as
\begin{align*}
\tau_i \geq \frac{m - r_{i \times n}([m]) }{n} = & \frac{m}{n} - \frac{1}{n} \min_{T \subseteq [m]} \Bigg( |[m] \setminus T| + n \cdot r_i(T) \Bigg) \tag{Theorem \ref{lemma:mat-union}} 
\end{align*}
We can now obtain the following lower bound by fixing a subset $\widetilde{T} \subseteq [m]$ in the right-hand-side of the expression above:  
\begin{align*}
\tau_i \geq \frac{m}{n} - \Bigg( \frac{|[m] \setminus \widetilde{T}|}{n} + r_i(\widetilde{T}) \Bigg).
\end{align*} 
Summing the above inequality for all agents $i \in [n]$, and then taking the maximum over all subsets $\widetilde{T}$ gives us

\begin{align*}
\sum_{i \in [n]} \tau_i & \geq \max_{\widetilde{T} \subseteq [m]} \Bigg\{ m - \Bigg( |[m] \setminus \widetilde{T}| + \sum_{i \in [n]} r_i(\widetilde{T}) \Bigg) \Bigg\} \\
& =   m - \min_{\widetilde{T} \subseteq [m]} \Bigg( |[m] \setminus \widetilde{T}| + \sum_{i \in [n]} r_i(\widetilde{T}) \Bigg) \\
& = m - \widehat{r}([m]) \tag{via Theorem \ref{lemma:mat-union}} \\
& = c^* \tag{via Lemma \ref{lemma:min-cost-characterization}}
\end{align*}
This gives us the desired inequality, $c^* \leq \sum_{i \in [n]} \tau_i$, and concludes the proof.
\end{proof}

With the previous lemmas established, we are now ready to present our algorithm that computes an MMS+PO allocation.

\floatname{algorithm}{Algorithm}
\begin{algorithm}[ht]
\caption{\textsc{MMSandEfficient}} \label{algo:PO-MMS}
\begin{tabularx}{\textwidth}{X}
{\bf Input:} Chore division instance  $\langle [n], [m], \{c_i\}_{i \in [n]}\rangle$ with value-oracle access to the costs $\{c_i\}_i$. \\
\textbf{Output:} Complete allocation $\Alloc = (A_1, \ldots, A_n)$.
\end{tabularx}
  \begin{algorithmic}[1]
  		\STATE Set $\left( A_1, \ldots, A_n \right) = \textsc{CostMinPartAlloc}(\langle [n], [m], \{c_i\}_{i} \rangle)$. \label{algPO+MMS:subroutine-call}
  		\STATE For each agent $i \in [n]$, set $\tau_i$ to be its minimax share. \COMMENT{The shares can be computed in polynomial time (Lemma \ref{lemma:minimax-share}).}
		\STATE Define set of unallocated chores $U \coloneqq  [m] \setminus \big( \cup_{i \in [n]} A_i \big)$. \label{line:una-mss}
  		 \FOR{each agent $i \in [n]$} \label{algPO+MMS:loop-begin}
		\WHILE{$c_i(A_i) < \tau_i$ and $U \neq \emptyset$} 
			\STATE Select any chore $t \in U$ and update $A_i \gets A_i + t$ along with $U \gets U - t$. \label{algPO+MMS:chore-assign}
			\ENDWHILE       
         \ENDFOR  \label{algPO+MMS:loop-end}
		\RETURN $\Alloc = \left( A_1, A_2, \ldots, A_n \right)$
		\end{algorithmic}
\end{algorithm}
In the following theorem, we show that Algorithm \ref{algo:PO-MMS} computes a social-cost-minimizing allocation that satisfies the minimax share guarantee. 
Algorithm \ref{algo:PO-MMS} first invokes Subroutine \ref{subroutine:cost-min-partial-alloc} to compute a partial allocation. Then, the algorithm completes the allocation by assigning the remaining chores sequentially to the agents. It performs this assignment without exceeding the maximin share $\tau_i$ (i.e., without violating the $\MMS$ guarantee) for any agent $i \in [n]$. In the theorem below we show that this simple process necessarily succeeds and, hence, finds an allocation that is both economically efficient and $\MMS$-fair.

\begin{restatable}{theorem}{apples}\label{thm:po-mms}
For binary supermodular cost functions, complete allocations that are $\MMS$-fair and social-cost minimizing always exist and can be computed in polynomial time via Algorithm \ref{algo:PO-MMS} (given value-oracle access to the costs).
\end{restatable}
\begin{proof}
We first show that Algorithm \ref{algo:PO-MMS} successfully returns a complete allocation $\Alloc$ that satisfies the minimax share guarantee. Note that for the partial allocation $\AllocB = (B_1, \ldots, B_n)$, returned by Subroutine \ref{subroutine:cost-min-partial-alloc} in Line \ref{algPO+MMS:subroutine-call}, we have 
\begin{align*}
c_i(B_i) = |B_i| - r_i(B_i) = |B_i| - |B_i| = 0 \quad \text{for each agent $i \in [n]$}.
\end{align*}
Recall that $r_i(B_i) = |B_i|$ since $B_i \in \mathcal{I}_i$ for all $i \in [n]$. 

After Line \ref{algPO+MMS:subroutine-call}, the chores that are left unallocated, $[m] \setminus \left( \cup_i B_i \right)$, are iteratively assigned in Lines \ref{algPO+MMS:loop-begin} to \ref{algPO+MMS:loop-end}. This assignment continues until $c_i(B_i) = \tau_i$ or no unallocated chores remain. Since the cost of every agent $c_i(B_i) = 0$, in the partial allocation $\AllocB$, and the costs bear the binary marginals property, each agent $i \in [n]$ can be additionally assigned at least $\tau_i$ chores while still maintaining $i$'s cost to be at most $\tau_i$. Therefore, collectively, all the agents can be assigned $\sum_{i \in [n]} \tau_i$ chores, while maintaining the inequality $c_i(A_i) \leq \tau_i$, for each agent $i \in [n]$. 

Furthermore, for the initial number of unassigned chores (in Line \ref{line:una-mss}) we have 
\begin{align*}
|U| = m - |\cup_{i \in [n]} B_i| = m - \widehat{r}([m]) = c^* \tag{via Lemma \ref{lemma:min-cost-characterization}}
\end{align*}
In addition, Lemma \ref{lemma:sum-of-shares} ensures that the number of chores that need to be assigned after Line \ref{line:una-mss} is upper bounded as follows: $|U| = c^* \leq \sum_{i \in [n]} \tau_i$. 

As mentioned previously, on top of $B_i$s, we can accommodate $\sum_{i \in [n]} \tau_i$ chores among the agents, while maintaining the $\MMS$ guarantee. Hence, all the chores remaining after Line \ref{line:una-mss} will be allocated in Lines \ref{algPO+MMS:loop-begin} to \ref{algPO+MMS:loop-end}, and the returned allocation $\Alloc$ will be $\MMS$-fair. 

For establishing the economic efficiency of the returned allocation $\Alloc = (A_1, \ldots, A_n)$, we invoke Lemma \ref{lemma:completing-the-partial-alloc}. In particular, for the returned allocation $\Alloc=(A_1, \ldots, A_n)$ we have $A_i \supseteq B_i$, for all agents $i \in [n]$, where $\AllocB = (B_1, \ldots, B_n)$ is the partial allocation returned by Subroutine \ref{subroutine:cost-min-partial-alloc} (Line \ref{algPO+MMS:subroutine-call}). Hence, via Lemma \ref{lemma:completing-the-partial-alloc}, we get that $\Alloc$ is social-cost minimizing (and, hence, Pareto efficient).

Finally, note that the minimax share $\tau_i$ of each agent $i \in [n]$ can be computed in polynomial time (given value oracles of the cost functions). This follows from Lemma \ref{lemma:minimax-share}. Also, Subroutine \ref{subroutine:cost-min-partial-alloc} can be executed in polynomial time. Therefore, overall, Algorithm \ref{algo:PO-MMS} finds---in the value-oracle model---the desired allocation in polynomial time. 

The theorem stands proved. 
\end{proof}

\subsection{Lorenz Dominating Allocations}
This section establishes that, for binary supermodular costs, Lorenz dominating allocations always exist and can be computed in polynomial time. Here, the existence follows from Theorem \ref{thm:lorenz} in which we show that Algorithm \ref{algo:Lorenz} necessarily succeeds in computing a Lorenz dominating allocation, under binary supermodular costs. 

Similar to the previous algorithms, Algorithm \ref{algo:Lorenz} starts by computing a partial allocation using Subroutine \ref{subroutine:cost-min-partial-alloc}. Following this, the remaining unallocated chores are split as equally as possible among the $n$ agents. In the following theorem, 
we establish that the resulting allocation is in fact Lorenz dominating.

\floatname{algorithm}{Algorithm}
\begin{algorithm}[ht]
\caption{\textsc{LorenzDominating}} \label{algo:Lorenz}
\begin{tabularx}{\textwidth}{X}
{\bf Input:} Chore division instance  $\langle [n], [m], \{c_i\}_{i \in [n]}\rangle$ with value-oracle access to the costs $\{c_i\}_i$. \\
\textbf{Output:} Complete allocation $\Alloc = (A_1, \ldots, A_n)$.
\end{tabularx}
  \begin{algorithmic}[1]
  		\STATE Set $\left( A_1, \ldots, A_n \right) = \textsc{CostMinPartAlloc}(\langle [n], [m], \{c_i\}_{i} \rangle)$. \label{algLorenz:subroutine-call}
  		\STATE Define set of unallocated chores $U \coloneqq  [m] \setminus \big( \cup_{i \in [n]} A_i \big)$. \label{line:una-lor}
  		\STATE Let index $h \coloneqq (|U| \mod{n})$ and, for agents $i \in \{1, 2, \ldots, h\}$, set integer $\alpha_i \coloneqq \left\lceil \frac{|U|}{n} \right\rceil$. \label{line:alpha-def-begin}
  		\STATE For the remaining agents, $i \in \{h+1, \ldots, n\}$, set $\alpha_i = \left\lfloor \frac{|U|}{n} \right\rfloor$. \label{line:alpha-def-end}
  		 \FOR{each agent $i \in [n]$} \label{algLorenz:loop-begin}
			\STATE Select any size-$\alpha_i$ subset of chores $T \subseteq U$ and update $A_i \gets A_i \cup T$ along with $U \gets U \setminus T$. \label{algLorenz:chore-assign}   
         \ENDFOR  \label{algLorenz:loop-end}
		\RETURN $\Alloc = \left( A_1, A_2, \ldots, A_n \right)$
		\end{algorithmic}
\end{algorithm}

\begin{restatable}{theorem}{lorenzthm}\label{thm:lorenz}
For binary supermodular cost functions, Lorenz dominating allocations always exist and can be computed in polynomial time via Algorithm \ref{algo:Lorenz} (in the value-oracle model).
\end{restatable}
\begin{proof}
By definition (see  Lines \ref{line:alpha-def-begin} and \ref{line:alpha-def-end} of the algorithm), the integers $\alpha_i$s satisfy  $\sum_{i \in [n]} \alpha_i = |U|$, where $U$ denotes the set of unassigned chores in Line \ref{line:una-lor}. This observation implies that $U$ can be partitioned in $n$ subsets of sizes $\alpha_1, \alpha_2, \ldots, \alpha_n$. Hence, the for-loop of the algorithm successfully assigns to each agent $i \in [n]$ a distinct subset of size $\alpha_i$ and the algorithm returns a complete allocation.  

Write $\Alloc = (A_1,\ldots, A_n)$ to denote the complete allocation returned by Algorithm \ref{algo:Lorenz}. The proof will be in two parts:
\begin{itemize}
\item[(a)] First, we will show that $c_i(A_i) = \alpha_i$, for each agent $i \in [n]$. Here, the integers $\alpha_i$s are as defined in Lines \ref{line:alpha-def-begin} and \ref{line:alpha-def-end} of the algorithm. 
\item[(b)] Subsequently, we will prove that $\Alloc$ is Lorenz dominating, i.e., show that the sorted cost profile of $\Alloc$ satisfies $\sigma(\Alloc) \geq_L \sigma(\mathcal{X})$, for every complete allocation $\mathcal{X} \in \Pi_n([m])$. 
\end{itemize} 

\noindent
\emph{Proving $(a)$}. Let $\AllocB = (B_1, \ldots, B_n)$ be the allocation returned by Subroutine \ref{subroutine:cost-min-partial-alloc} in Line \ref{algLorenz:subroutine-call}. Since the allocation $\Alloc$ returned by the algorithm satisfies $A_i \supseteq B_i$, for all agents $i \in [n]$, Lemma \ref{lemma:completing-the-partial-alloc} ensures that $\Alloc$ is a social-cost-minimizing allocation. In particular, $\sum_{i \in [n]} c_i(A_i) = c^*$. Furthermore, the social cost of the partial allocation $\AllocB$ (returned by Subroutine \ref{subroutine:cost-min-partial-alloc}) is equal to $0$ -- this follows from the fact that $c_i(B_i) = |B_i| - r_i(B_i) = |B_i| - |B_i| = 0$, for all agents $i \in [n]$. Hence, the difference in social cost of $\Alloc$ and $\AllocB$ can be bounded as 
\begin{align*}
\sum_{i \in [n]} c_i(A_i) - \sum_{i \in [n]} c_i(B_i) & = c^* - 0 \\ 
& = m - \widehat{r}([m]) \tag{via Lemma \ref{lemma:min-cost-characterization}} \\
& = m - |\cup_{i \in [n]} B_i| \tag{$(B_1, \ldots, B_n)$ is a basis of $\umat$} \\
& = |U|
\end{align*}
Here, $U$ denotes the set of unassigned chores in Line \ref{line:una-lor}, i.e., $U = [m] \setminus \big( \cup_{i \in [n]} B_i \big)$. As mentioned previously, we have $\sum_{i \in [n]} \alpha_i = |U|$. Therefore, 
\begin{align}
\sum_{i \in [n]} c_i(A_i) - \sum_{i \in [n]} c_i(B_i) = \sum_{i \in [n]} c_i(A_i) - 0 = \sum_{i \in [n]} \alpha_i \label{ineq:sum-alpha}	
\end{align}
However, $A_i$ was obtained by assigning a size-$\alpha_i$ subset of chores to agent $i \in [n]$; see Line \ref{algLorenz:chore-assign}. Therefore, the binary marginals property of the cost functions, imply $c_i(A_i) \leq c_i(B_i) + \alpha_i = \alpha_i$, for all agents $i \in [n]$. These upper bounds and inequality (\ref{ineq:sum-alpha}) give us the desired equality $c_i(A_i) = \alpha_i$, for all agents $i \in [n]$. \\

\noindent
\emph{Proving $(b)$}. Using $(a)$, we can write the sorted cost profile of the returned allocation $\Alloc$ as follows: $\sigma(\Alloc) = (\alpha_1, \alpha_2, \ldots, \alpha_n)$ where $\alpha_i = \lceil |U|/n \rceil$ for $i \leq h$ and $\alpha_i = \lfloor |U|/n \rfloor$ for $i > h$. Consider any complete allocation $\AllocX \in \Pi_n([m])$ and write its sorted cost profile $\sigma(\AllocX) = (\beta_1, \beta_2, \ldots, \beta_n)$. Here, the costs are sorted, $\beta_1 \geq \ldots \geq \beta_n$, and their sum is at least $c^*$ (the minimum social cost of the instance), i.e., $\sum_{i=1}^n \beta_i \geq c^*$. We will show that, for all indices $\ell \in [n]$, the following inequality holds: $\sum_{i=1}^\ell \alpha_i \leq \sum_{i=1}^\ell \beta_i$. Therefore, we will have $\sigma(\Alloc) \geq_L \sigma(\AllocX)$ for all complete allocation $\AllocX$. 

Assume, towards a contradiction, that $\sum_{i=1}^\ell \alpha_i > \sum_{i=1}^\ell \beta_i$ for some index $\ell \in [n]$. We can assume, without loss of generality, that $\ell$ is the smallest index for which this inequality holds; otherwise, we can reduce $\ell$ and maintain the prefix-sum inequality. Given that $\sum_{i=1}^\ell \alpha_i > \sum_{i=1}^\ell \beta_i$ and $\ell$ is the smallest inequality for which this inequality holds, it must be the case that $\alpha_\ell > \beta_\ell$. Furthermore, since $\alpha_\ell$ is either $\lceil |U|/n \rceil$ or $\lfloor |U|/n \rfloor$, we get that $\beta_\ell \leq \lfloor |U|/n \rfloor$. Furthermore, the sorting of the $\beta_i$s imply $\beta_j  \leq \lfloor |U|/n \rfloor$ for all indices $j \geq \ell$. These bounds, however, contradict the fact that $\sum_{i=1}^n \beta_i \geq c^*$; specifically, 
\begin{align*}
c^* - \sum_{i=1}^n \beta_i & = \sum_{i=1}^n \alpha_i - \sum_{i=1}^n \beta_i \\
& = \left(\sum_{i=1}^\ell \alpha_i - \sum_{i=1}^\ell \beta_i \right) + \sum_{j=\ell+1}^n \left(\alpha_j -  \beta_j \right) \\
& > 0 + \sum_{j=\ell+1}^n \left(\alpha_j -  \beta_j \right) \tag{by definition of $\ell$} \\ 
& \geq 0 \tag{since $\beta_j \leq \lfloor |U|/n \rfloor \leq \alpha_j$, for all $j \geq \ell$}
\end{align*}
The last inequality simplifies to $c^* >  \sum_{i=1}^n \beta_i$ and, hence, contradicts the definition of $c^*$. Therefore, it must be the case that, for all indices $\ell \in [n]$, we have $\sum_{i=1}^\ell \alpha_i \leq \sum_{i=1}^\ell \beta_i$. This shows that $\Alloc$ is a Lorenz dominating allocation and establishes (b).

Finally, note that Subroutine \ref{subroutine:cost-min-partial-alloc}, Lines \ref{algLorenz:loop-begin} to \ref{algLorenz:loop-end}, and, hence, Algorithm \ref{algo:Lorenz} can be executed in polynomial time, given value-oracle access to the cost functions. This completes the proof. 
\end{proof}

\section{Incomparability of Notions} 
\label{sec:incom}


Given the encompassing positive results obtained in the previous section for binary supermodular costs, one might wonder whether the three considered fairness desiderata (namely, Lorenz domination, $\MMS$, and $\EFone$) are implied by each other, under this class of functions. This section shows that this is not the case, i.e., there exist allocations that satisfy one of these criteria, but not any other. That is, the notions are \emph{incomparable} in the context of binary supermodular costs and an allocation that satisfies any one of these criteria does not necessarily satisfy any other. 

We also show the following stronger result: there exist instances (with binary supermodular costs) in which Lorenz domination is \textit{incompatible} with $\MMS$ and $\EFone$. Specifically, the set of Lorenz dominating allocations is disjoint from the sets of $\MMS$ and $\EFone$ allocations. The result implies that, under binary supermodular costs, it is impossible to obtain Lorenz domination alongside $\MMS$ or $\EFone$. Notably, this result is in contrast to the case of goods, where such implications often hold. For example, in the fair division of goods and under the complementary class of binary submodular valuations, any Lorenz dominating allocation is always $\EFone$~\cite{babaioff2021fair}.

\begin{theorem} \label{thm:incomparable}
For binary supermodular cost functions, (i) Lorenz domination, (ii) the $\MMS$ guarantee, and (iii) $\EFone$ are incomparable.
\end{theorem}
\begin{proof}
We will first show that Lorenz domination and MMS are incomparable. Consider the following chore division instance with $m=11$ chores and $n=3$ agents each with a binary supermodular cost function. Agent~1's cost is simply the cardinality, $c_1(S) = |S|$, for any subset $S \subseteq [m]$. Agent~2 and agent~3 have the following (identical) cost function: $c_2(S) = c_3(S) = \max\{ 0, |S|-3 \}$. One can verify that all the three cost functions are binary supermodular. 

This instance admits a Lorenz dominating cost profile of $(2,2,1)$ induced by assigning $1$ chore to agent $1$ (and, hence, her cost is $1$) and $5$ chores each to agents 2 and 3 (their costs are equal to $2$).\footnote{Alternatively, agent~1 can receive a cost of $2$ and one of the other agents receives $4$ chores and, hence, incurs a cost of $1$.} 

Note that, in this instance, agent~1's minmax share $\tau_1 = 4$, while agents~2 and~3 have shares $\tau_2 = \tau_3 =1$. Hence, in order for agents~2 and~3 to achieve the $\MMS$ guarantee, they must each get at most four items. However, in this case, agent~1 must get at least three items with a cost of $3$. Therefore, the cost profile of any $\MMS$-fair allocation is $(3,1,1)$ (or worse, in a Lorenz domination sense). Since the cost profile $(3,1,1)$  is Lorentz dominated by $(2,2,1)$, in this instance $\MMS$ does not imply Lorenz domination. 

Similarly, in any Lorenz dominating allocation at least two agents have a cost of $2$. Hence, either agent~2 or agent~3 does not achieve its $\MMS$ guarantee (i.e., receives a bundle of cost more than its minmax share).

The same instance shows that Lorenz domination and $\EFone$ are incomparable. In any $\EFone$ allocation, the $m=11$ chores have to be partitioned into bundles of (near-equal) sizes of $4$, $4$ and $3$. Hence, in any $\EFone$ allocation agent~1 necessarily obtains a cost of at least $3$. That is, an $\EFone$ allocation cannot be Lorentz dominating.  

Complementarily, in any allocation that is Lorenz dominating, agent~1 must receive at most $2$ chores, while at least one of the remaining agents receives $5$ chores and envies agent~1, beyond the removal of one chore. 

Finally, we will show that $\EFone$ and $\MMS$ are incomparable. Consider the following instance, with $m=10$ chores and $n=3$ agents with binary supermodular costs. In particular, agent 1's and agent 2's cost for a set of chores is the cardinality, $c_1(S) = c_2(S) = |S|$, for any subset $S \subseteq [m]$. Agent~3 always obtains a marginal cost of $1$ for chores $t_1, t_2, t_3$, and $t_4$, but sees the pairs $\{t_5, t_6\}$, $\{t_7, t_8\}$ and $\{t_9, t_{10}\}$ as complements. That is, agent~3 only obtains a cost increase of $1$ whenever it is allocated both chores out of a pair. For the first two agents, the minimax shares  $\tau_1 = \tau_2 = 4$. For agent 3, the minimax share $\tau_3 = 2$; this is obtained via a partition in which all the three pairs are separated and, hence, add zero cost to any bundle for agent 3.  
 
The specific allocation $\Alloc=(A_1, A_2, A_3)$, with $A_1 = \{t_4, t_5, t_6\}$, $A_2 = \{t_7, t_8, t_9, t_{10}\}$, and $A_3 = \{t_1, t_2, t_3\}$, is $\EFone$ but not $\MMS$, since agent~3 has a cost of $3$. On the other hand, the allocation $\AllocB=(B_1, B_2, B_3)$, with bundles $B_1 = \{t_1, t_2, t_5\}$, $B_2 = \{t_6, t_8, t_{10}\}$, $B_3 = \{t_3, t_4, t_7, t_9\}$ is $\MMS$-fair but not $\EFone$: agent~3 envies agent~2 even after the removal of any chore from the bundle of agent~3. Hence, $\EFone$ and $\MMS$ are incomparable, and neither property implies the other for an allocation.
\end{proof}

\noindent \textbf{Remark.} The first example in the above proof shows, in fact, the stronger result stated previously: for this instance, the set of Lorenz dominating allocations is disjoint from the sets of $\MMS$ and $\EFone$ allocations. Hence, Lorenz domination and $\EFone$---along with Lorenz domination and $\MMS$---are incompatible.




\section{$\EFX$ Chore Allocations} \label{sec:efx}
In the context of binary supermodular costs, we have established the existence of Pareto efficient and $\EFone$ allocations. This motivates us to consider the stronger fairness criterion of $\EFX$. In stark contrast to the results of Section \ref{sec:positive}, we show that for binary supermodular cost functions, Pareto efficient allocations that are even \emph{approximately} $\EFX$ do not exist. The notion of approximate $\EFX$ is defined as follows.

%
%
%
%
%

\begin{definition}[Approximate $\EFkX$]
For any $\beta \in [0,1]$ and integer $k \geq 1$, in a chore division instance $\langle [n], [m], \{c_i\}_{i \in [n]} \rangle$, an allocation $\Alloc = (A_1, A_2, \ldots, A_n)$ is said to be $\beta$-$\EFkX$ iff for all pair of agents $i,j \in [n]$, with $|A_i| \geq k$, we have $\beta c_i(A_i \setminus T) \leq c_i(A_j)$ for all size-$k$ subsets $T \subseteq A_i$.
\end{definition}

An allocation is said to be $\EFX$ iff it satisfies the definition above with $\beta = 1$ and $k = 1$. We now state establish the main negative result for this section. 

\begin{restatable}{theorem}{negativepoefx}\label{theorem:negative-PO-EFX}
Under binary supermodular costs, for any $\beta \in (0,1]$ and any integer $k\geq 1$, allocations that are both Pareto efficient and $\beta$-$\EFkX$ do not necessarily exist. Furthermore, the nonexistence holds even in the case of identical (binary supermodular) cost functions.
\end{restatable}

\begin{proof}
Consider the following chore division instance with $m=2k+1$ chores and $n=2$ agents that have identical binary supermodular cost functions. The two agents always obtain a marginal cost of $1$ for chore $t_{2k+1}$, but see (chore) pairs $\{t_1, t_2\}, \{t_3, t_4\}, \ldots, \{t_{2k-1}, t_{2k}\}$ as complements. That is,  each agent only experiences a cost increase of $1$ when it is allocated both the chores from a pair, and has a marginal cost of $0$ when it is allocated only one chore from a pair. 

In any Pareto optimal allocation, the agents must incur no cost for the $k$ pairs of items. Hence, in any PO (complete) allocation each of the $k$ pairs are split between the two agents (with each agent getting one chore), and one agent gets the final chore $t_{2k+1}$. One can verify that, in the instance at hand, no other allocation is Pareto optimal. Consequently, $(0,1)$ and $(1,0)$ are the only cost profiles under all PO allocations. 

However, in any PO allocation, the agent that receives the chore $t_{2k+1}$ envies the other agent even after the removal of the other $k$ items from its bundle. Since the other agent has a cost of zero, the remaining envy is not eliminated via scaling down by any multiplicative factor. Therefore, in the given instance, for any $\beta\in(0,1]$ and any positive integer $k \geq 1$, there does not exist a complete allocation that is both Pareto efficient and $\beta$-$\EFkX$.
\end{proof}



\subsection{Finding $\EFX$ Chore Allocations}

Complementing the previous negative result (Theorem \ref{theorem:negative-PO-EFX}) and focusing on fairness alone, we show that under \textit{identical monotone} cost functions, an $\EFX$ allocation always exists. Notably, this result only requires the (identical) cost function to be monotonic and is obtained via an algorithmic approach. In fact, for cost functions $c$ that are integer-valued (i.e., $c(S) \in \mathbb{Z}_{\geq 0}$ for all subsets $S$), our algorithm (Add-and-Fix) runs in pseudo-polynomial time (Theorem \ref{theorem:identical-costs}). Hence, for the particular case of identical cost functions with binary marginals (which are not required to be supermodular), we obtain a polynomial-time algorithm for finding $\EFX$ allocations. 

Our positive results for $\EFX$ can be directly adapted to the goods settings. Thus we obtain a pseudo-polynomial time algorithm for finding $\EFX$ allocations of goods under identical, monotone, and integer-valued valuations.

We obtain this positive result via Algorithm \ref{algo:EFX} (Add-and-Fix), which executes with value-oracle access to the cost function. Add-and-Fix incrementally constructs an $\EFX$ allocation (in the outer while loop, Lines \ref{algEFX:outer-loop-begin}-\ref{algEFX:inner-loop-end}), always maintaining that the current partial allocation is $\EFX$. In each incremental step, the algorithm updates the bundle of an agent $\ell \in [n]$ with the lowest current cost. In particular, first the algorithm adds a currently unassigned chore $t$ into agent $\ell$'s bundle (i.e., into $A_\ell$). Then, complementing the chore inclusion, the algorithm removes chores from $A_\ell$ (see the while-loop at Line \ref{algEFX:inner-loop-begin}) until $\EFX$ is restored with respect to the agent $s \in [n]$ with the second-lowest cost. 

We will show that the (partial) allocation obtained after each such update is also $\EFX$. Moreover, each update either increases the social cost, or decreases the number of unallocated chores (while maintaining the social cost). This property is used to argue that Add-and-Fix terminates and the allocation obtained at the end is $\EFX$. The result is formally stated in the following theorem.

\floatname{algorithm}{Algorithm}
\begin{algorithm}[ht]
\caption{\textsc{Add-and-Fix}} \label{algo:EFX}
{\bf Input:} Chore division instance $\langle [n], [m], \{c \} \rangle$ with value-oracle access to the identical cost $c$. \\
\textbf{Output:} Complete allocation $\Alloc = (A_1, \ldots, A_n)$.
\begin{algorithmic}[1]
\STATE Initialize set of unassigned chores $U = [m]$ and bundles $A_i = \emptyset$ for each agent $i \in [n]$.
\WHILE{$U \neq \emptyset$} \label{algEFX:outer-loop-begin}
\STATE Select any chore $t \in U$. \label{algEFX:chore-t} 
\STATE Select $\ell \in \argmin_{k \in [n]} c(A_k)$. \label{algEFX:min-agent}
\STATE Select $s \in \argmin_{k \in [n] \setminus \{ \ell \}} c(A_k)$. \label{algEFX:sec-min-agent}
\STATE Update bundle $A_{\ell}  \gets A_\ell + t$ and the set of unassigned chores $U \gets U - t$ \label{algEFX:t-addition}  \label{algEFX:t-removal}
			\WHILE{there exists a chore $q \in A_\ell$ such that $c(A_\ell - q) > c(A_s)$} \label{algEFX:inner-loop-begin} 
				\STATE Update $A_\ell \gets A_\ell - q$ and $U \gets U + q$. \label{algEFX:q-removal} \label{algEFX:q-addition}
			\ENDWHILE \label{algEFX:inner-loop-end}
			
		\ENDWHILE \label{algEFX:outer-loop-end}
		\STATE \textbf{return } $\Alloc = \left( A_1, A_2, \ldots, A_n \right)$
		\end{algorithmic}
\end{algorithm}

\begin{restatable}{theorem}{idencosts}
\label{theorem:identical-costs}
In any chore division instance in which the agents have identical costs $c$ with binary marginals, an $\EFX$ allocation always exists and can be computed by Add-and-Fix (Algorithm \ref{algo:EFX}) in polynomial time (given value-oracle access to the cost function).
\end{restatable}
\begin{proof}
We begin by showing that Add-and-Fix maintains the $\EFX$ guarantee as an invariant, i.e., each partial allocation $\Alloc$ considered at the beginning of the outer-while loop is $\EFX$. 


We provide an inductive argument showing that the $\EFX$ invariant is maintained throughout the algorithm's execution. In particular, assume that the partial allocation $\Alloc$ at the start of an iteration of the outer-while loop is $\EFX$ (which is trivially true for the first iteration). We will show that the (partial) allocation at the end of the iteration, which we will denote by $\Alloc'$, will also be $\EFX$. 

Note that, in Line \ref{algEFX:t-addition}, an unallocated chore $t \in U$ is included in the bundle of the agent $\ell \in [n]$ with the lowest cost bundle. If, after this assignment, the (partial) allocation continues to be $\EFX$, then the inner-while loop will not execute and the allocation at the end of the iteration, $\Alloc'$, will be $\EFX$ (i.e., the invariant is maintained). Otherwise, the addition of chore $t$ violates $\EFX$, which can only happen if the cost of agent $\ell$'s bundle increases and agent $\ell$ starts envying agent $s$ (i.e., $c(A_\ell) > c(A_s)$). In this case, the inner while-loop will iteratively continue to remove chores $q$ from agent $\ell$'s bundle as long as $c(A_\ell - q) > c(A_s)$. Note that the selection criterion ensures that even after the removal of $q$, the cost of agent $\ell$'s bundle is more than $c(A_s)$. Hence, the inner-while loop continues until we obtain an allocation $\Alloc'=(A'_1, \ldots, A'_n)$ satisfying $c(A'_\ell) > c(A'_s)$ and $c(A'_\ell - q) \leq c(A'_s)$ for all chores $q \in A'_\ell$. 

This resultant allocation $\Alloc'$ is $\EFX$ for agent $\ell$, since, for all chores $q \in A'_i$ and all agents $k \in [n] \setminus \{\ell\}$, we have 
\begin{align*}
c(A'_\ell - q) \leq c(A'_s) = c(A_s) \leq c(A_k) = c(A'_k).
\end{align*} 
In addition, $\Alloc'$ will be $\EFX$ for all other agents $k \in [n] \setminus \{\ell\}$: the bundles $A'_k = A_k$, and the cost of agent $\ell$'s bundle increases, $c(A'_\ell) > c(A'_s) = c(A_s) \geq c(A_k - q) = c(A'_k - q)$, for all chores $q \in A'_k$. This establishes that the (partial) allocation $\Alloc'$ is $\EFX$. 

We will next show that the outer while-loop successfully assigns all the chores within polynomially many iterations. That is, the outer while-loop terminates with a complete $\EFX$ allocation. Add-and-Fix thus necessarily finds an $\EFX$ allocation in polynomial time.   

Note that if, in an iteration of the outer while-loop, the assignment of chore $t$ (to agent $\ell$) does not lead to an $\EFX$ violation, then the number of unassigned chores, $|U|$, decreases in that iteration. Otherwise, if the chore assignment violates $\EFX$, then social cost among the agents strictly increases: as observed above, in such a case we end up at an allocation $\Alloc'$ with the property that $c(A'_\ell) > c(A_s) \geq c(A_\ell)$ and $c(A'_k) = c(A_k)$, for all agents $k \in [n] \setminus \{\ell\}$. Hence, here, the cumulative cost of the agents strictly increases. Now, given that the cost function has binary marginals, the social cost can only increase $m$ times and $|U|$ can decrease consecutively at most $m$ times. Hence, the outer while-loop will terminate in at most $m^2$ iterations and provide a complete $\EFX$ allocation. Also, note that we only require value-oracle access to the cost function to implement all the steps in the algorithm. 

This, overall, shows that the algorithm returns an $\EFX$ allocation in polynomial time.
\end{proof}

As shown above, Algorithm \ref{algo:EFX} terminates in polynomial time when the identical cost function has binary marginals. It is relevant to note further that Add-and-Fix will always terminate---albeit not necessarily in polynomial time---as long as the cost function is monotone. Since the proof of Theorem \ref{theorem:identical-costs} implies that the returned allocation is $\EFX$, under any identical monotone cost function, we obtain the following corollary.

\begin{corollary}
$\EFX$ chore allocations always exist under identical monotone cost functions. 
\end{corollary}

This leads to an alternate proof for the existence of $\EFX$ allocations in chore division instances with identical cost functions. By extending a similar result for the goods setting by Plaut and Roughgarden~\cite{plaut2020almost} to chore division, Aleksandrov~\cite{aleksandrov2018almost} showed that for identical monotone cost functions any allocation satisfying the so called \textit{leximax-} notion is $\EFX$. However, as we show in Appendix~\ref{appendix:leximax-hardness}, computing a leximax- allocation is {\rm NP}-hard even when the cost functions have binary marginals. Bypassing the leximax- approach and the associated {\rm NP}-hardness, Add-and-Fix finds an $\EFX$ allocation in polynomial time for this setting.

\bibliographystyle{alpha}
\bibliography{references}

\newpage

\appendix
\section{Missing Proofs from Section \ref{sec:cost-min}} 
\label{appendix:missing-proofs-sc}

In this section, we first state and prove a useful proposition. Then, we establish Lemma \ref{lemma:submod-supermod}.
\begin{proposition}
\label{proposition:property-bin-super}
If $f:2^{[m]} \mapsto \mathbb{R}_{\geq 0}$ is a binary supermodular function and $f(S) > 0$ for some subset $S \subseteq [m]$, then there always exists an element $a \in S$ such that $f(S - a) = f(S) - 1$.
\end{proposition}
\begin{proof}
Since the function $f$ is supermodular, we know that the inequality 

\begin{align}
\label{prop1-eq1}
f(T+a) - f(T) \leq f(S) - f(S - a)
\end{align}
must hold for any $a \in S$ and subset $T \subseteq (S-a)$. Now, consider any ordering $s_1, s_2, \ldots, s_k$ of all the elements in the set $S$. Write $S_i \coloneq \{s_1, \ldots, s_i\}$ as the set of first $i$ elements in this ordering, for any index $i \in [k]$, and let $S_0 \coloneqq \emptyset$ be the empty set. Given this, we can instantiate equation (\ref{prop1-eq1}) with $T = S_{i-1}$ and $a = s_i$ to get $f(S) - f(S - s_i) \geq f(S_i) - f(S_{i-1})$ for all $i \in [k]$. Summing, we obtain 
\begin{align*}
\sum_{i \in [k]} \left( f(S) - f(S - s_i) \right) & \geq \sum_{i \in [k]} \left( f(S_i) - f(S_{i-1}) \right) \\ 
& = f(S_k) - f(S_0) \tag{telescoping sum} \\ 
& = f(S) \tag{since $S_k=S$ and $S_0 = \emptyset$}  \\
& > 0 \tag{Lemma assumption}
\end{align*}
Therefore, there must exist $s_i \in S$ such that $f(S) - f(S - s_i) >0$. The fact that function $f$ has binary marginals implies that $f(S) - f(S - s_i) = 1$ and completes the proof.
\end{proof}

\SubmodSupermod*
\begin{proof}
We will show that $(i)$ function $f$ has binary marginals iff $g$ has binary marginals and $(ii)$ $f$ is supermodular iff $g$ is submodular. Together $(i)$ and $(ii)$ establish the lemma.\footnote{Recall that a set function $g$ is a matroid-rank function iff it is binary submodular~\cite{schrijver2003combinatorial}.} \\

\noindent
\emph{Part $(i)$.} Since the function $f$ bears the marginals property, we have $f(S + a) - f(S) \in \{0, 1\}$, for all subsets $S \subseteq [m]$ and $a \in [m]$. Therefore, $g(S+a) - g(S) = |S+a| - f(S+a) - (|S| - f(S)) = 1 - (f(S+a) - f(S)) \in \{0,1\}$. Hence, the marginals of function $g$ are binary as well. The proof of the other direction (i.e., if $g$ is dichotomous then so is $f$) follows symmetrically, since $f(S) = |S| - g(S)$. \\

\noindent
\emph{Part $(ii)$.} By definition, the supermodularity of $f$ implies $f(S + a) - f(S) \leq f(T+a) -f(T)$, for all subsets $S \subseteq T$ and $ a \in [m] \setminus T$. By substituting $f(S) = |S| - g(S)$, we obtain $(|S+a| - g(S + a)) - (|S| - g(S)) \leq (|T+a| - g(T + a)) - (|T| - g(T))$. Simplifying the last inequality gives us $1 - \left( g(S+a) - g(S) \right) \leq 1 - \left( g(T+a) - g(T) \right)$. That is, $g(S+a) - g(S) \geq g(T+a) - g(T)$. Hence, $g$ is submodular. A symmetric argument shows that the reverse direction holds, i.e., if $g$ is submodular then $f$ is supermodular. 

The lemma stands proved. 
\end{proof}

\section{Missing Proof from Section \ref{section:mms-po}}
\label{appendix:proofs-positive}

The section restates and proves Lemma \ref{lemma:minimax-share}.

\minimaxshare*
\begin{proof}
Fix any agent $i \in [n]$ and consider a chore division instance $\mathcal{I}$ over the $m$ chores and $n$ agents each with identical cost function $c_i$. Lemma \ref{lemma:min-cost-characterization} implies that, for this instance $\mathcal{I}$, the minimum social cost $c_\mathcal{I}^* = m - r_{i \times n}([m])$. Here, we use the observation that if all the $n$ agents have the same cost function $c_i$,  then $\widehat{r}([m]) = r_{i \times n}([m])$. Let $(M_1, \ldots, M_n)$ denote a complete allocation that induces the minimax share $\tau_i$ (see equation (\ref{eq:defn-minimax})); in particular, $\tau_i = \max_j c_i(M_j)$. Using an averaging argument, we get 
\begin{align*}
\tau_i & \geq \left\lceil \frac{1}{n} \sum_{j \in [n]} c_i(M_j) \right\rceil \\ 
& \geq \left\lceil \frac{c^*_\mathcal{I}}{n} \right\rceil \tag{via the definition of $c^*_\mathcal{I}$} \\ 
& = \left\lceil \frac{m - r_{i \times n}([m])}{n} \right\rceil \numberthis \label{ineq:taui-lb}
\end{align*}

In fact, we can show that $\tau_i = \left\lceil \frac{m - r_{i \times n}([m])}{n} \right\rceil$. Towards this we first establish the claim below. 

\begin{claim}
\label{claim:interim}
In instance $\mathcal{I}$ there necessarily exists a complete allocation $\Alloc^* = (A^*_1, \ldots, A^*_n)$ with the properties that 
\begin{itemize}
\item Allocation $\Alloc^*$ minimizes the social cost, $\sum_{p \in [n]} c_i(A^*_p) = c^*_\mathcal{I}$, and 
\item Under $c_i$, the costs of any two bundles in $\Alloc^*$ differ by at most $1$, i.e., $c_i(A^*_j) \leq c_i(A^*_k) + 1$, for all indices $j, k \in [n]$. 
\end{itemize}
\end{claim}

\begin{proof}
Consider any social-cost-minimizing allocation $\Alloc = (A_1, \ldots, A_n)$ that violates the second property, i.e., for $\Alloc$ there exist two indices $j, k \in [n]$ such that $c_i(A_j) > c_i(A_k) + 1$. Then, we can update the allocation $\Alloc$ to obtain another (complete) allocation $\Alloc'= (A'_1, \ldots, A'_n)$ such that $\sum_{p \in [n]} c_i(A'_p)^2 <  \sum_{p \in [n]} c_i(A_p)^2$ and $\Alloc'$ continues to be social-cost minimizing. In particular, using the fact that $c_i(A_j) > c_i(A_k) + 1 \geq 0$, we can invoke Proposition \ref{proposition:property-bin-super} (stated and proved in Appendix \ref{appendix:missing-proofs-sc}) to infer the existence of a chore $t \in A_j$ such that $c_i(A_j - t) = c_i(A_j) - 1$. We transfer the chore $t$ from $A_j$ to $A_k$ and obtain the modified allocation $\Alloc' = (A'_1, \ldots, A'_n)$, i.e., we set $A'_j = A_j - t$ and $A'_k = A_k + t$ along with $A'_p = A_p$, for all other indices $p \in [n] \setminus \{j, k\}$. We show that allocation $\Alloc'$ continues to minimize social cost and satisfies $\sum_{p \in [n]} c_i(A'_p)^2 <  \sum_{p \in [n]} c_i(A_p)^2$. 

To show that $\Alloc'$ is social-cost-minimizing, note that $c_i(A'_j) = c_i(A_j) - 1$ and $c_i(A'_k) \leq c_i(A_k) + 1$; recall that the cost function $c_i$ has binary marginals. Furthermore, $c_i(A'_p) = c_i(A_p)$, for all other indices $p \in [n] \setminus \{j, k\}$. These observation imply that the social cost of allocation $\Alloc'$ is at most the social cost of $\Alloc$. Since $\Alloc$ is a social-cost-minimizing allocation, we get that the updated allocation $\Alloc'$ is also social-cost-minimizing. Next, note that the difference $\sum_{p \in [n]} c_i(A'_p)^2 - \sum_{p \in [n]} c_i(A_p)^2$ is at most $\big( (c_i(A_j) - 1)^2 + (c_i(A_k) + 1)^2 \big) - \big( c_i(A_j)^2 + c_i(A_k)^2 \big) = 2 (c_i(A_k) + 1 - c_i(A_j) ) < 0$; the last inequality follows from the claim assumption that $c_i(A_k) + 1 < c_i(A_j)$. 

If we iteratively keep on performing the above-mentioned update, then this process must terminate. In fact, it does so in polynomially many iterations, since the potential $\sum_{p \in [n]} c_i(A_p)^2$ is bounded from above by $nm^2$ and below by $0$; note that $\sum_{p \in [n]} c_i(A_p)^2 \geq 0$ and each update decreases the potential by at least one. Furthermore, when the process terminates, we obtain the desired allocation $\Alloc^* = (A^*_1, \ldots, A^*_n)$. This follows from the fact that the update process maintained social-cost-minimization as an invariant, and the process stopped only when the costs, under $c_i$ of any two bundles differed by at most $1$, i.e., we had $c_i(A^*_j) \leq c_i(A^*_k) + 1$, for all indices $j, k \in [n]$. The claim stands proved. 
\end{proof}

As mentioned previously, we will use Claim \ref{claim:interim} to show that $\tau_i = \left\lceil \frac{m - r_{i \times n}([m])}{n} \right\rceil$. The two properties satisfied by the allocation $\Alloc^*$ (in Claim \ref{claim:interim}) will imply that the highest cost among the bundles $\max_{p \in [n]} c_i(A^*_p) = \left\lceil \frac{c^*_\mathcal{I}}{n} \right\rceil$. Considering the specific complete allocation $\Alloc^*$ and definition of $\tau_i$ gives us 
\begin{align}
\tau_i \leq  \left\lceil \frac{c^*_\mathcal{I}}{n} \right\rceil  = \left\lceil \frac{m - r_{i \times n}([m])}{n} \right\rceil \label{ineq:taui-ub}
\end{align}
The complementary inequalities (\ref{ineq:taui-lb}) and (\ref{ineq:taui-ub}) lead to the stated bound $\tau_i = \left\lceil \frac{m - r_{i \times n}([m]) }{n} \right\rceil$. 

Finally, note that, using the matroid union algorithm, we can compute $r_{i \times n}([m])$ in polynomial time. Hence, the minimax shares $\tau_i$s are polynomial-time computable in the current context. This completes the proof of the lemma. 
\end{proof}

\section{{\rm NP}-Hardness of Leximax- for Identical Costs with Binary Marginals}
\label{appendix:leximax-hardness}

Given a chore allocation $\Alloc = (A_1, \ldots, A_n)$, we can define a \emph{cost-size} tuple, $(c(A_i), |A_i|)$, for each agent $i \in [n]$. Furthermore, we can define a total order $>_{lex}$ over the cost-size tuples as follows: $(c(A_i), |A_i|) >_{lex} (c(A_j), |A_j|)$ iff $c(A_i) > c(A_j)$, or $c(A_i) = c(A_j)$ and $|A_i| > |A_j|$. An allocation $\Alloc^*$ is leximax- iff it maximizes the minimum cost-size tuple among the agents (with respect to the total order $>_{lex}$), subject to that, maximizes the second minimum cost-size tuple among the agents, and so on.\footnote{The definition of leximax- is the exact analog of leximin++. The term leximin++ has been used in the context of good division ~\cite{plaut2020almost}, whereas the leximax- construct has been used in chore division ~\cite{aleksandrov2018almost}.}

\begin{theorem}
Computing leximax- allocations is {\rm NP}-Hard even under identical cost functions with binary marginals. 
\end{theorem}
\begin{proof}
To prove this, we reduce the well-known NP-Hard problem of exact $3$-cover to the problem of computing leximax- allocations for identical cost functions with binary marginals. Recall that, an instance of exact $3$-cover is a tuple $(X, \mathcal{F})$, where $X$ is a universe of elements with size $|X| = 3q$, for some integer $q \geq 1$, and $\mathcal{F} \subseteq 2^X$ is a collection of subsets of $X$. Each given subset $F \in \mathcal{F}$ is of three, i.e., $|F| = 3$. In the exact $3$-cover problem, for any given the instance $(X, \mathcal{F})$, the goal is to decide whether there exists a sub-collection $\mathcal{F}' \subseteq \mathcal{F}$ with the properties that $(i)$ all elements of $X$ are covered by $\mathcal{F}'$, i.e., $\cup_{F \in \mathcal{F}'} F = X$, and $(ii)$ no element is covered twice, i.e., $F \cap G = \emptyset$ for all subsets $F, G \in \mathcal{F}'$, with $F \neq G$. Equivalently, the goal here is to determine if there exists a size-$q$ collection $\mathcal{F}' \subseteq \mathcal{F}$ that exactly covers $X$. 

Given an exact $3$-cover instance $(X,\mathcal{F})$, with $|X| = 3q$, we will, in polynomial time, construct a chore division instance $\mathcal{I}$ in which the agents have identical cost functions (with binary marginals). Specifically, we construct an instance with $n=q$ agents and a chore corresponding to each element of $X$, i.e., $X$ denotes the set of chores. The identical cost function of the agents $c: 2^X \mapsto \mathbb{R}_{\geq 0}$ is defined as $c(S) = \max_{F \in \mathcal{F}} \ |S \cap F|$, for all subsets $S \subseteq X$. Note that the cost function $c$ has binary marginals, and $c(S) \leq 3$ for all $S \subseteq X$. 

To complete the proof, we will show that the instance $(X, \mathcal{F})$ has an exact $3$-cover iff, in the the constructed chore division instance $\mathcal{I}$, any leximax- allocation $\Alloc = (A_1, A_2, \allowbreak \ldots, A_n)$ satisfies $c(A_i) = 3$, for all $i \in [n]$. For the forward direction, consider the case in which the given cover instance admits an exact $3$-cover $\mathcal{F}' \subseteq \mathcal{F}$. Write $\mathcal{F}' = \{S_1, S_2, \ldots, S_q\}$ and define allocation $\mathcal{S} = (S_1, S_2, \ldots, S_q)$. Note that, by definition, the cost of each agent $i \in [n]$ will be $c(S_i) = \max_{S \in \mathcal{F}} |S_i \cap S| = |S_i \cap S_i| = |S_i| = 3$; additionally, $c(S) \leq 3$ for all $S \subseteq X$. Since in allocation $\Alloc$ the cost of each agent is $3$, which is the maximum possible cost for the cost function $c$, all leximax- allocations (that maximize the minimum cost-size tuple, then the second minimum cost-size tuple and so on) must also allocate bundles of cost $3$ to all agents. That is, for any leximax- allocation $\Alloc = (A_1, A_2, \ldots, A_n)$ of the chore division instance $\mathcal{I}$ we have $c(A_i) = 3$ for all agents $i \in [n]$. 

For proving the reverse direction, we will show that if $\Alloc = (A_1, A_2, \ldots, A_n)$ is a leximax- allocation such that $c(A_i) = 3$, for all $i \in [n]$, then the given instance $(X, \mathcal{F})$ admits an exact cover. Towards this, note that by definition of cost function $c$, if for all agents $i \in [n]$, $c(A_i) = \max_{F \in \mathcal{F}} |A_i \cap F| = 3$ then there exists subset $F_i \in \mathcal{F}$ such that $F_i \subseteq A_i$. Since, the bundles of all agents are disjoint, the subsets  $F_i \subseteq A_i$ (which satisfy $F_i \in \mathcal{F}$) are also disjoint. However, this implies that, together, the subsets $\{F_1, F_2, \ldots, F_q\}$ cover all the elements of $X$: since $|F_i| = 3$ we have $|\cup_{i \in [n]} F_i| = \sum_{i \in [n]} |F_i| = 3n = 3q = |X|$. Additionally, no element is covered twice since $F_i \cap F_j = \emptyset$, for $i \neq j$. That is, $\{F_1, F_2, \ldots, F_q\}$ forms an exact cover for the given cover instance. This completes the reduction and establishes that computing leximax- allocations for instances with binary marginals is as hard as solving exact $3$-cover, i.e., leximax- computation is $\mathrm{NP}$-hard.
\end{proof}

\end{document}